\documentclass[a4paper,onecolumn,11pt,accepted=2023-03-31]{quantumarticle}
\pdfoutput=1

\usepackage[english]{babel}
\usepackage[T1]{fontenc}

\usepackage{lipsum}
\usepackage[utf8]{inputenc}
\usepackage{CJKutf8}

\usepackage{amsthm}
\usepackage{amsmath}
\usepackage{amssymb}
\usepackage{mathrsfs}
\usepackage{bbm}
\usepackage{bm}
\usepackage{quantikz}
\usepackage{caption}
\usepackage{subcaption}

\usepackage{hyperref}

\usepackage{algorithm}
\usepackage{algpseudocode}

\floatname{algorithm}{Protocol}

\DeclareMathOperator{\Tr}{Tr}

\renewcommand{\proj}[1]{| #1 \rangle\!\langle #1 |}
\renewcommand{\braket}[2]{\langle #1 | #2 \rangle}
\renewcommand{\ket}[1]{|#1 \rangle}

\newtheorem{theorem}{Theorem}
\newtheorem{corollary}[theorem]{Corollary}
\newtheorem{lemma}[theorem]{Lemma}
\newtheorem{definition}[theorem]{Definition}

\newtheorem{claim}[theorem]{Claim}
\newtheorem{remark}[theorem]{Remark}

\def \ketbra#1#2{|#1\rangle\! \langle #2|}
\def \sandwich#1#2#3{\langle#1|#2|#3\rangle}
\def \KeyGen {\normalfont{ \texttt{KeyGen}}}
\def \Enc {\normalfont{ \texttt{Enc}}}
\def \Eval {\normalfont{ \texttt{Eval}}}
\def \Dec {\normalfont{ \texttt{Dec}}}
\def \CL {\normalfont{ \texttt{CL}}}
\def \CNOT {\normalfont{\texttt{CNOT}}}
\def \XGate {\normalfont{\texttt{X}}}

\def \ZGate {\normalfont{\texttt{Z}}}
\def \A {\normalfont{\textsf{A}}}
\def \O {\normalfont{\textsf{O}}}
\def \R {\normalfont{\textsf{R}}}
\def \B {\normalfont{\textsf{B}}}
\def \X {\normalfont{\textsf{X}}}
\def \Y {\normalfont{\textsf{Y}}}

\def \E {\normalfont{\textsf{E}}}
\def \Accept {\normalfont{\texttt{Accept}}}
\def \Abort {\normalfont{\texttt{Abort}}}

\makeatother

\usepackage[normalem]{ulem} 
\usepackage{color}

\newcommand\ylhS[1]{[YH:{\let\helpcmd\sout\parhelp#1\par\relax\relax}] }
\long\def\parhelp#1\par#2\relax{%
	\helpcmd{#1}\ifx\relax#2\else\par\parhelp#2\relax\fi%
}

\begin{document}
\begin{CJK*}{UTF8}{gbsn}
\title{Privacy and correctness trade-offs for information-theoretically secure quantum homomorphic encryption}

\author{Yanglin Hu (胡杨林)}
\orcid{0009-0000-1105-589X}
\affiliation{ Centre of Quantum Technologies, National University of Singapore, Singapore}

\author{Yingkai Ouyang}  
\orcid{0000-0003-1115-0074}
\affiliation{ 
Centre of Quantum Technologies, National University of Singapore, Singapore}

\author{Marco Tomamichel}  
\orcid{0000-0003-1115-0074}
\affiliation{ 
Centre of Quantum Technologies, National University of Singapore, Singapore}
\affiliation{Department of Electrical and Computer Engineering, National University of Singapore}

\maketitle
\end{CJK*}

\begin{abstract} 
    Quantum homomorphic encryption, which allows computation by a server directly on encrypted data, is a fundamental primitive out of which more complex quantum cryptography protocols can be built. For such constructions to be possible, quantum homomorphic encryption must satisfy two privacy properties: data privacy which ensures that the input data is private from the server, and circuit privacy which ensures that the ciphertext after the computation does not reveal any additional information about the circuit used to perform it, beyond the output of the computation itself. While circuit privacy is well-studied in classical cryptography and many homomorphic encryption schemes can be equipped with it, its quantum analogue has received little attention. Here we establish a definition of circuit privacy for quantum homomorphic encryption with information-theoretic security. Furthermore, we reduce quantum oblivious transfer to quantum homomorphic encryption. By using this reduction, our work unravels fundamental trade-offs between circuit privacy, data privacy and correctness for a broad family of quantum homomorphic encryption protocols, including schemes that allow only the computation of Clifford circuits.
\end{abstract}

\section{Introduction} 

Given the difficulty of building reliable quantum computers at scale, it is reasonable to expect the first such quantum computers to be controlled by a few service providers with the prerequisite infrastructure and resources~\cite{Fitzsimons_2017}. In such a scenario, individual users will, in contrast, only hold quantum computers with far more limited capabilities and must delegate complex computations to large servers. However, there is also an inherent lack of trust between the service providers and the users: individuals would like to keep their data private from large corporations, while the service providers would like to keep their exact implementation private from competitors. Protocols such as blind quantum computing~\cite{Aharonov_2008,Broadbent_2009,Morimae_2013,Reichardt_2013,Fitzsimons_2017,Mantri_2017} are proposed to help with the situation; they require only minimal quantum capabilities of the client but require extensive communication. Quantum homomorphic encryption~\cite{Yu_2014,Broadbent_2015,Dulek_2016,Tan_2016,Ouyang_2018,Mahadev_2020,Ouyang_2022}, which allows a server to compute on encrypted data of a client without first decrypting it, offers an alternative solution to this problem. Here, all communication can be done in one round, but the client needs a quantum computer for encoding the input and decoding the output.

Because classical homomorphic encryption~\cite{Gentry_2009,Gentry_2009_thesis,Gentry_2010} can build a broad range of more complex classical cryptographic primitives such as multiparty secure computation and private information retrieval, it has been called the ``Swiss army knife'' of classical cryptography~\cite{Barak_2012,Lindell_2017,Esmaeilzade_2022}. These reductions\footnote[1]{Suppose that we have two protocols, protocol $A$ and protocol $B$. If we can use protocol $B$ to realise protocol $A$, then we can reduce protocol $A$ to protocol $B$~\cite{Reingold_2004}. If protocol $A$ is impossible, then protocol $B$ is also impossible. Thus limitations of protocol $A$ also apply to protocol $B$. } rely crucially on the data and circuit privacy of homomorphic encryption schemes. While data privacy is inherent in homomorphic encryption, circuit privacy encapsulates the property that no additional information about the computation circuit is leaked beyond its action on the encoded data.

It is natural to conjecture a quantum analogue of the classical reduction, i.e. reducing quantum oblivious transfer to quantum homomorphic encryption. In this paper, we progress in understanding the extent to which quantum homomorphic encryption can be analogously a ``Swiss army knife'' of quantum cryptography. In particular, we explore the limitations of quantum homomorphic encryption in the paradigm of information-theoretic security.

Quantum homomorphic encryption that allows the delegation of an {\em arbitrary} quantum computation while also assuring the privacy of the client's encrypted data cannot exist because its existence would violate well-known information-theoretic bounds, such as Nielsen's no-programming bound~\cite{Yu_2014} or fundamental coding-type bounds~\cite{Lai_2018,Newman_2018} such as Nayak's bound~\cite{Nayak_1999} which bounds the amount of classical information that can be stored in a quantum state. 
On the other hand, if we consider quantum homomorphic encryption schemes that support only Clifford computations, such schemes exist with asymptotically perfect data privacy and correctness~\cite{Ouyang_2018}. 
Hence, there is the hope that, by restricting ourselves to quantum homomorphic encryption schemes that support a limited set of operations~\cite{Tan_2016,Tan_2017,Ouyang_2018,Lai_2018,Ouyang_2020,Ouyang_2022}, such schemes can still have enough functionality (data privacy, correctness and circuit privacy) to serve as a Swiss army knife. Here we show that this is not possible. 
Namely, even if we restrict ourselves to quantum homomorphic encryption protocols that can perform only two-qubit Clifford gates, we find that data privacy, circuit privacy and correctness cannot be simultaneously achieved. 
In particular, we obtain non-trivial trade-offs between these parameters for such computationally-restricted quantum homomorphic encryption schemes.

Our main contributions are as follows:
\begin{itemize}
    \item We introduce a formal definition of circuit privacy for quantum homomorphic encryption schemes. (See Definition~\ref{def:circuit_privacy}.) 
    We achieve this by introducing a quantum counterpart of the simulation paradigm~\cite{Lindell_2017} in classical cryptography. Roughly speaking, the simulation paradigm is a pattern of using a simulator to compare a possibly insecure actual protocol with a naturally secure ideal protocol.

    \item We give an explicit reduction from quantum oblivious transfer to quantum homomorphic encryption by constructing a quantum oblivious transfer protocol with a quantum homomorphic encryption protocol. In this reduction, we use only quantum homomorphic encryption protocols that perform delegated Clifford circuits and additionally utilize genuine random classical bits. (See Theorem~\ref{thm:construction}).
    \item The reduction allows us to inherit no-go results for quantum oblivious transfer~\cite{Chailloux_2013,Chailloux_2016,Amiri_2021} to quantum homomorphic encryption. We find that, for any information-theoretically secure quantum homomorphic encryption scheme support (at least) Clifford operations, it holds that
    \begin{align}
        \label{eq:nogo-bound}
            \epsilon_d + \epsilon_c + 4\sqrt{\epsilon} \geq \frac{1}{2} ,
    \end{align}
    where $\epsilon_d$, $\epsilon_c$ and $\epsilon$ are parameters describing data privacy, circuit privacy and correctness, respectively, and ideally, we would want them to all be small.
\end{itemize}

It is worth emphasizing that our results apply only to quantum homomorphic encryption schemes with information-theoretic security.
In particular, quantum homomorphic encryption schemes based on computational hardness assumptions~\cite{Broadbent_2015,Dulek_2016,Mahadev_2020} might be able to support better trade-offs for these parameters. Notably, circuit privacy for computational-secure semi-honest quantum homomorphic encryption was discussed in~\cite{Dulek_2016}.

Our paper is structured as follows. In Section~\ref{sec:definitions}, we give formal definitions of quantum cryptographic primitives. In Section~\ref{subsec:definitions_of_qhe}, we give the scheme of quantum homomorphic encryption and define the correctness, data privacy and circuit privacy. In Section~\ref{subsec:definitions_of_qot}, we discuss two types of quantum oblivious transfer, standard and semi-random oblivious transfer, and show their equivalence. In Section~\ref{sec:reduction}, we reduce standard oblivious transfer to quantum homomorphic encryption. In Section~\ref{sec:bounds_for_qot}, we present a bound for semi-random oblivious transfer. In Section~\ref{sec:bounds_for_qhe}, we present bounds for quantum homomorphic encryption. In Section~\ref{sec:lower_bounds_for_qhe}, we obtain our lower bound for quantum homomorphic encryption by reduction. In Section~\ref{sec:upper_bounds_for_qhe}, we derive our upper bound.

\section{Quantum homomorphic encryption and oblivious transfer}\label{sec:definitions}

\subsection{Notations}

When $x\in \{0,1\}$ denotes a bit, we let $\overline{x}=x \oplus 1\in \{0,1\}$ denote its complement. An $n$-bit string is a binary vector $(x_1,...,x_n)\in \{0,1\}^n$. A random bit is denoted by $\$ $. We use Latin capital letters in a sans serif font, such as $\X$, to denote the system and also its Hilbert space. The set of density matrices of a system $\X$ is denoted by $\mathscr{S}(\X)$. The set of completely positive trace-preserving maps from a system $\X$ to a system $\Y$ is denoted by ${\rm CPTP}(\X, \Y)$. The set of unitary channels on $\X$ is denoted by $\mathscr{U}(\X)$. The adjoint channel $\mathcal{N}^*\in{\rm CPTP}(\Y,\X)$ of a channel $\mathcal{N}\in{\rm CPTP}(\X,\Y)$ is defined by the relation $\Tr(\rho \mathcal{N}(\sigma))=\Tr(\mathcal{N}^*(\rho)\sigma)$ for all $\rho\in\mathscr{S}(\Y)$ and $\sigma\in\mathscr{S}(\X)$. The identity channel of $n$ dimensions is denoted by $\mathcal{I}_n$. The identity channel on $\X$ is denoted by $\mathcal{I}_{\X}$. For simplicity, we will omit identity channels $\mathcal{I}_{\X}$ and $\mathcal{I}_{\X}$ if it does not cause any confusion. The set of unitary operators on $\X $ is denoted by ${\rm U}(\X)$. The identity operator acting on $n$ dimensions and on $\X$ are denoted by $\mathbbm{I}_n$ and  $\mathbbm{I}_{\X}$, respectively. The Schatten $1$-norm is defined by $\|A\|_1 = \sqrt{A^\dagger A}$. The trace distance is defined by $\Delta(A,B)=\frac{1}{2}\|A-B\|_1$. The fidelity is defined by $F(A,B)=\|\sqrt{A}\sqrt{B}\|_1$. The Hermitian conjugate of a term in an equation is denoted by a $\textnormal{h.c.}$ following the term. Table~\ref{table:notations} summarizes the notations. 
\begin{table}[!htbp]
    \centering
    \begin{tabular}{c|c}
        \hline
        Notation & Description \\
        \hline
        $\$$ & A random bit. \\
        $\X$ & The system and also its Hilbert space. \\
        $\mathscr{S}(\X)$ & The set of density matrices on $\X$. \\
        ${\rm CPTP}(\X,\Y)$ & The set of CPTP maps from $\X$ to $\Y$. \\
        $\mathscr{U}(\X)$ & The set of unitary channels on $\X$. \\
        $\mathcal{N}^*$ & The adjoint channel of $\mathcal{N}$. \\
        $\mathcal{I}_n$ & The identity channel of $n$ dimensions. \\
        $\mathcal{I}_{\X}$ & The identity channel on $\X$. \\
        ${\rm U}(\X)$ & The set of unitary matrices on $\X$. \\
        $\mathbbm{I}_n$ & The identity operator of $n$ dimensions. \\
        $\mathbbm{I}_{\X}$ & The identity operator on $\X$. \\
        $\|\,\cdot\,\|_1$ & The Schatten $1$-norm. \\ 
        $\Delta(\,\cdot\,,\,\cdot\,)$ & The trace distance between two states. \\ 
        $F(\,\cdot\,,\,\cdot\,)$ & The fidelity between two states. \\
        $\textnormal{h.c.}$ & The Hermitian conjugate of a term. \\
        \hline
    \end{tabular}
    \caption{Notations}
    \label{table:notations}
\end{table}

\subsection{Classical homomorphic encryption}
For the formal definition of classical homomorphic encryption, interested readers may refer to~\cite{Gentry_2009_thesis, Lindell_2017}. Here we only introduce classical homomorphic encryption informally, emphasising its scheme and circuit privacy. In a classical homomorphic encryption protocol, Alice, the user, encrypts the input and sends the ciphertext to Bob, the server. Next, Bob evaluates a function on the ciphertext without decryption and sends the evaluated ciphertext back to Alice. Finally, Alice decrypts the evaluated ciphertext and obtains the output. 

The correctness of classical homomorphic encryption requires that the output is the same as the function computed on the input. Data privacy requires that Bob cannot distinguish the ciphertexts corresponding to different inputs. 

Circuit privacy is defined in a simulation paradigm. It is defined by comparing a possibly insecure actual protocol to a naturally secure ideal protocol. Alice knows the input, the ciphertext, and the modified ciphertext in the actual protocol. In the ideal protocol, Alice knows the input and the function computed by the input. Bob's circuit is private if Alice does not learn more information about the circuit in the actual protocol than in the ideal protocol. This happens when a simulator can simulate the results of the actual protocol with the results of the ideal protocol. Hence, circuit privacy quantifies the simulator's performance.

In classical cryptography, we can reduce classical oblivious transfer to classical homomorphic encryption. For a generic but simple classical reduction from classical oblivious transfer to classical homomorphic encryption, interested readers may refer to~\cite{Esmaeilzade_2022}. The essence of the reduction is that classical homomorphic encryption ensures not only data privacy but also circuit privacy. We will use this idea in the quantum analogue of classical reduction. 

\subsection{Quantum homomorphic encryption}\label{subsec:definitions_of_qhe}

In this subsection, we introduce relevant definitions of quantum homomorphic encryption. The scheme of quantum homomorphic encryption is similar to that of classical homomorphic encryption, which is presented in Definition~\ref{def:qhe_scheme}. Both data privacy and circuit privacy are essential if we treat quantum homomorphic encryption as a ``Swiss army knife'' primitive. We formally define correctness, data privacy and circuit privacy as a quantum analogue of their classical counterparts in Definition~\ref{def:correctness}, Definition~\ref{def:data_privacy} and Definition~\ref{def:circuit_privacy}, respectively.

Figure~\ref{fig:qhe_scheme} describes a quantum homomorphic encryption scheme. In a quantum homomorphic encryption protocol, Alice computes a function $\KeyGen$ to obtain the key, uses the key and encryption map $\Enc$ to encrypt Alice's input and then sends the encryption to Bob. Bob applies $\Eval$ to evaluate Bob's channel $\mathcal{F}$ on the encryption and sends the evaluated encryption back to Alice. Alice decrypts the evaluated encryption with the key and a decryption map $\Dec$ and obtains Alice's output.  
\begin{figure}[!htpb]
    \centering
    \includegraphics[scale=0.35]{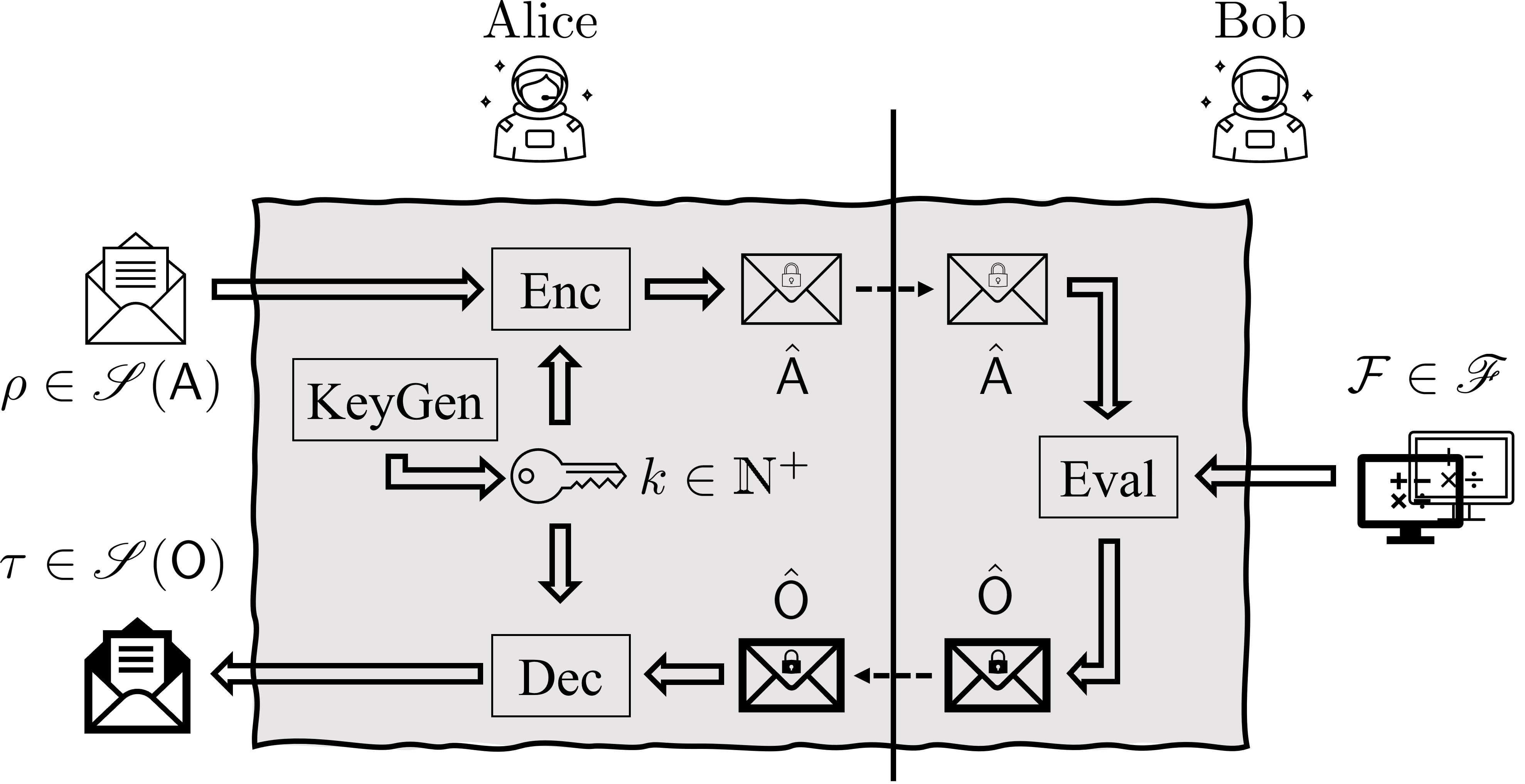}
    \caption{The scheme of quantum homomorphic encryption. }
    \label{fig:qhe_scheme}
\end{figure} 
\begin{definition}[Scheme for quantum homomorphic encryption]\label{def:qhe_scheme}
    A quantum homomorphic encryption protocol is a secure two-party computation protocol where Alice requires Bob to evaluate a channel in one round of communication. It is characterized by a set of channels and four maps $(\mathscr{F}, \KeyGen,\Enc, \Eval, \Dec)$.
    
    \begin{itemize}
        \item $\mathscr{F}\subset {\rm CPTP}(\A,\O)$ is a set of channels that Alice would like to delegate.
        \item $\KeyGen: \emptyset\rightarrow \{0,1\}^{L}$ is the key generation map which generates a random key $k\in \{0,1\}^L$ where $L$ is the length of the key. Without loss of generality, we define the space of all keys as the space of all $L$-bit strings $\{0,1\}^{L}$. 
        \item $\Enc:\{0,1\}^{L} \times\mathscr{S}(\A) \rightarrow \mathscr{S}(\hat{\A})$ is the encryption map which takes the key $k$ and Alice's unencrypted input $\rho\in \mathscr{S}(\A)$ and gives out Alice's encrypted message $\sigma\in \mathscr{S}(\hat{\A})$. 
        \item $\Eval:\mathscr{F}\times  \mathscr{S}(\hat{\A})\rightarrow \mathscr{S}(\hat{\O})$ is the evaluation map which takes in a channel $\mathcal{F}\in \mathscr{F}$ and Alice's encrypted message $\sigma\in \mathscr{S}(\hat{\A})$ and gives out Bob's encrypted message $\theta\in \mathscr{S}({\hat{\O}})$.
        \item $\Dec: \{0,1\}^{L}\times \mathscr{S}(\hat{\O})\rightarrow  \mathscr{S}(\O)$ is the decryption map which takes in the key $k$ and Bob's encrypted message $\theta\in \mathscr{S}({\hat{\O}})$ and gives out Alice's unencrypted output $\tau\in \mathscr{S}(\O)$.
    \end{itemize}
    For simplicity, we will denote $\Enc(k,\,\cdot\,)$ by $\Enc_k$, $\Eval(\mathcal{F},\,\cdot\,)$ by $\hat{\mathcal{F}}(\,\cdot\,)$, $\{\Eval(\mathcal{F},\,\cdot\,): \mathcal{F}\in\mathscr{F}\}$ by $\hat{\mathscr{F}}$ and $\Dec(k,\,\cdot\,)$ by $\Dec_k(\,\cdot\,)$. 
\end{definition}

\begin{table}[!htbp]
    \centering
    \begin{tabular}{c|c|c}
        \hline
        Function & Input & Output\\
        \hline
        $\KeyGen$ & $\emptyset$ & $k\in\{0,1\}^{L}$ \\
        $\Enc$ & $k\in \{0,1\}^L$, $\rho\in \mathscr{S}(\A)$ & $\sigma\in\mathscr{S}(\hat{\A})$ \\
        $\Enc_k$ & $\rho\in \mathscr{S}(\A)$ & $\sigma\in\mathscr{S}(\hat{\A})$ \\
        $\Eval$ & $\mathcal{F}\in \mathscr{F}$, $\sigma\in\mathscr{S}(\hat{\A})$ & $\theta\in\mathscr{S}(\hat{\O})$ \\
        $\hat{\mathcal{F}}$ & $\sigma\in\mathscr{S}(\hat{\A})$ & $\theta\in\mathscr{S}(\hat{\O})$ \\
        $\Dec$ & $k\in\{0,1\}^L$, $\theta\in\mathscr{S}(\hat{\O})$ & $\tau\in\mathscr{S}(\O)$ \\ 
        $\Dec_k$ & $\theta\in\mathscr{S}(\hat{\O})$ & $\tau\in\mathscr{S}(\O)$ \\ 
        \hline
    \end{tabular}
    \caption{channels in the quantum homomorphic encryption protocol}
\end{table}

For a quantum homomorphic encryption protocol to be meaningful, it needs to return the correct output. This is quantified by a property known as correctness.

\begin{definition}[Correctness of quantum homomorphic encryption]\label{def:correctness}
    A quantum homomorphic encryption protocol $(\mathscr{F}, \KeyGen, \Enc, \Eval, \Dec)$ is $\epsilon$-correct if, for every channel $\mathcal{F}\in \mathscr{F}$, every Alice's input and its purification $\ket{\psi'}\in {\A \R_{\A}}$ and every key $k\in\{0,1\}^L$, 
    \begin{align}
        \Delta\left(((\Dec_k \hat{\mathcal{F}} \Enc_k)\otimes \mathcal{I}_{\R_{\A}})[\proj{\psi'}_{\A \R_{\A}}] , (\mathcal{F}\otimes\mathcal{I}_{\R_{\A}})(\proj{\psi'}_{\A \R_{\A}})\right)\leq \epsilon. 
    \end{align}
    Specifically when $\epsilon=0$, we say that the protocol is perfectly correct.
\end{definition}

Quantum homomorphic encryption needs to protect an honest Alice's data when a malicious Bob strives to learn it, as is shown in Figure~\ref{fig:qhe_cheating_bob}. A quantum homomorphic encryption protocol is data private if a malicious Bob cannot distinguish different inputs of an honest Alice. The trace distance between encrypted states describes their indistinguishability.

\begin{figure}[!htpb]
    \centering
    \includegraphics[scale=0.35]{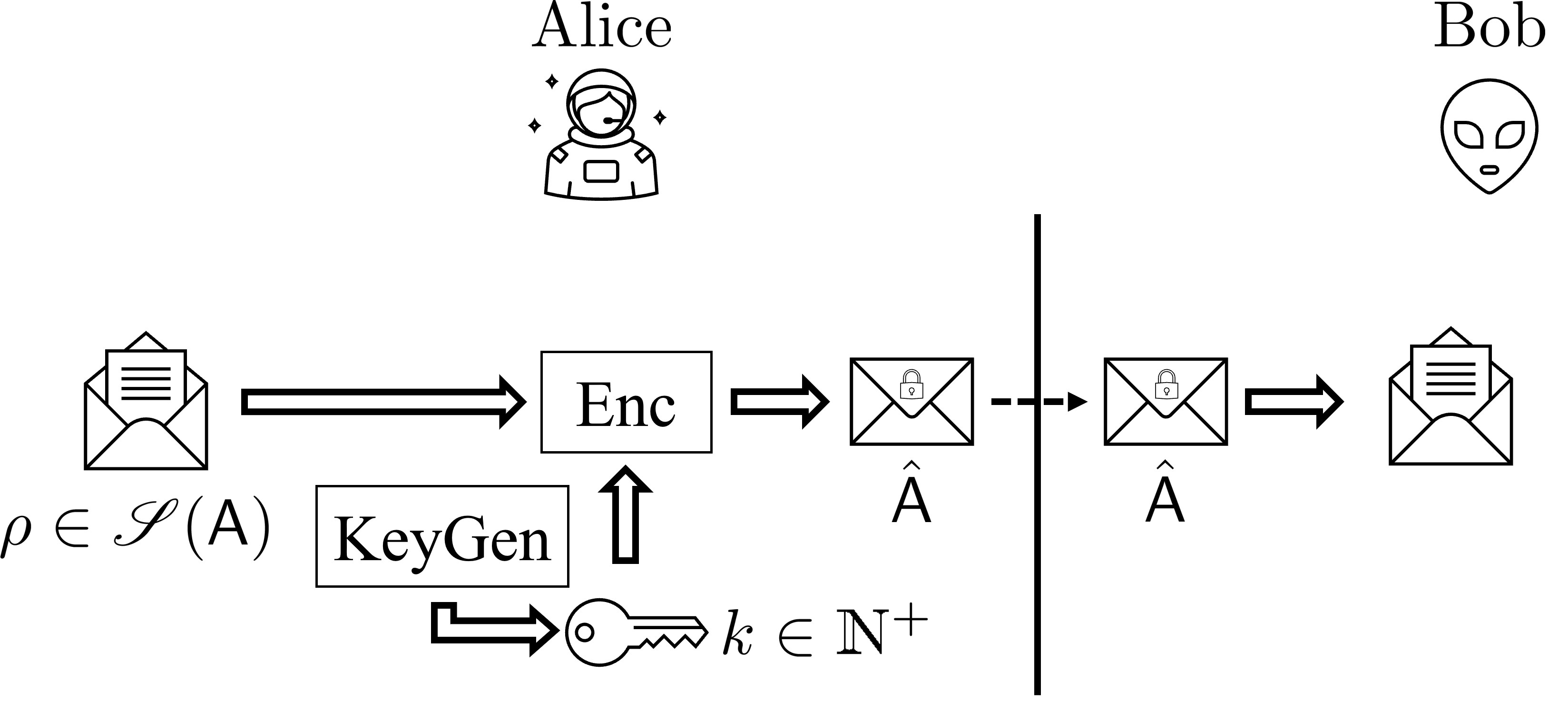}
    \caption{An honest Alice and malicious Bob in quantum homomorphic encryption}
    \label{fig:qhe_cheating_bob}
\end{figure}

\begin{definition}[Information-theoretic data privacy of quantum homomorphic encryption]\label{def:data_privacy}
    A quantum homomorphic encryption protocol $(\mathscr{F}, \KeyGen,\Enc, \Eval, \Dec)$ is $\epsilon_d$-data private if for any two Alice's inputs $\rho\in \mathscr{S}(\A)$ and $\rho'\in\mathscr{S}(\A)$, the trace distance between the averages of Alice messages $\mathbbm{E}(\Enc_k[\rho])$ and $\mathbbm{E}(\Enc_k[\rho'])$ over the distribution over keys $k\in\{0,1\}^L$ is bounded by $\epsilon_d$ 
    \begin{align}
        \Delta(\mathbbm{E}(\Enc_k[\rho]),\mathbbm{E}(\Enc_k[\rho']))\leq \epsilon_d, 
    \end{align}
    Specifically when $\epsilon_d=0$, we say that the protocol is perfectly data private. 
\end{definition}

Quantum homomorphic encryption needs to protect an honest Bob's circuit when a malicious Alice strives to learn Bob's circuit more than what an honest Alice's output indicates, as is shown in Figure~\ref{fig:qhe_cheating_alice}. However, even an honest Alice can learn some information about Bob's circuit. Therefore, we identify the information indicated by an honest Alice's output with an ideal protocol in Figure~\ref{fig:qhe_ideal}. In the ideal protocol, we imagine a trusted Charlie. Alice sends the input, and Bob sends the channel to Charlie. Charlie applies Bob's channel to Alice's input and sends the output to Alice. Alice ought to learn no more information in the actual protocol than what Alice can learn in the ideal protocol in a circuit private protocol. This circuit private case happens if a channel can turn Alice's output in the ideal protocol into Alice's output in the actual protocol. The trace distance between the two states can quantify the channel's performance. 

Now we describe the actual protocol and the ideal protocol in detail. In the actual protocol, Alice sends her message $\sigma\in \mathscr{S}({\hat{\A}})$ (whose purification is $\ket{\psi}\in {\hat{\A}\R_{\hat{\A}}}$) to Bob. Bob applies $\hat{\mathcal{F}}\in \hat{\mathscr{F}}$ on $\sigma$ according to Bob's channel $\mathcal{F}$, and sends Bob's message $\hat{\mathcal{F}}(\sigma)\in\mathscr{S}({\hat{\O}})$ to Alice. Alice finally possesses the joint state $\hat{\mathcal{F}}[\proj{\psi}]\in \mathscr{S}({\hat{\O}\R_{\hat{\A}}})$ of Bob's message and Alice's referencing system. 

\begin{figure}[!htpb]
    \centering
    \includegraphics[scale=0.35]{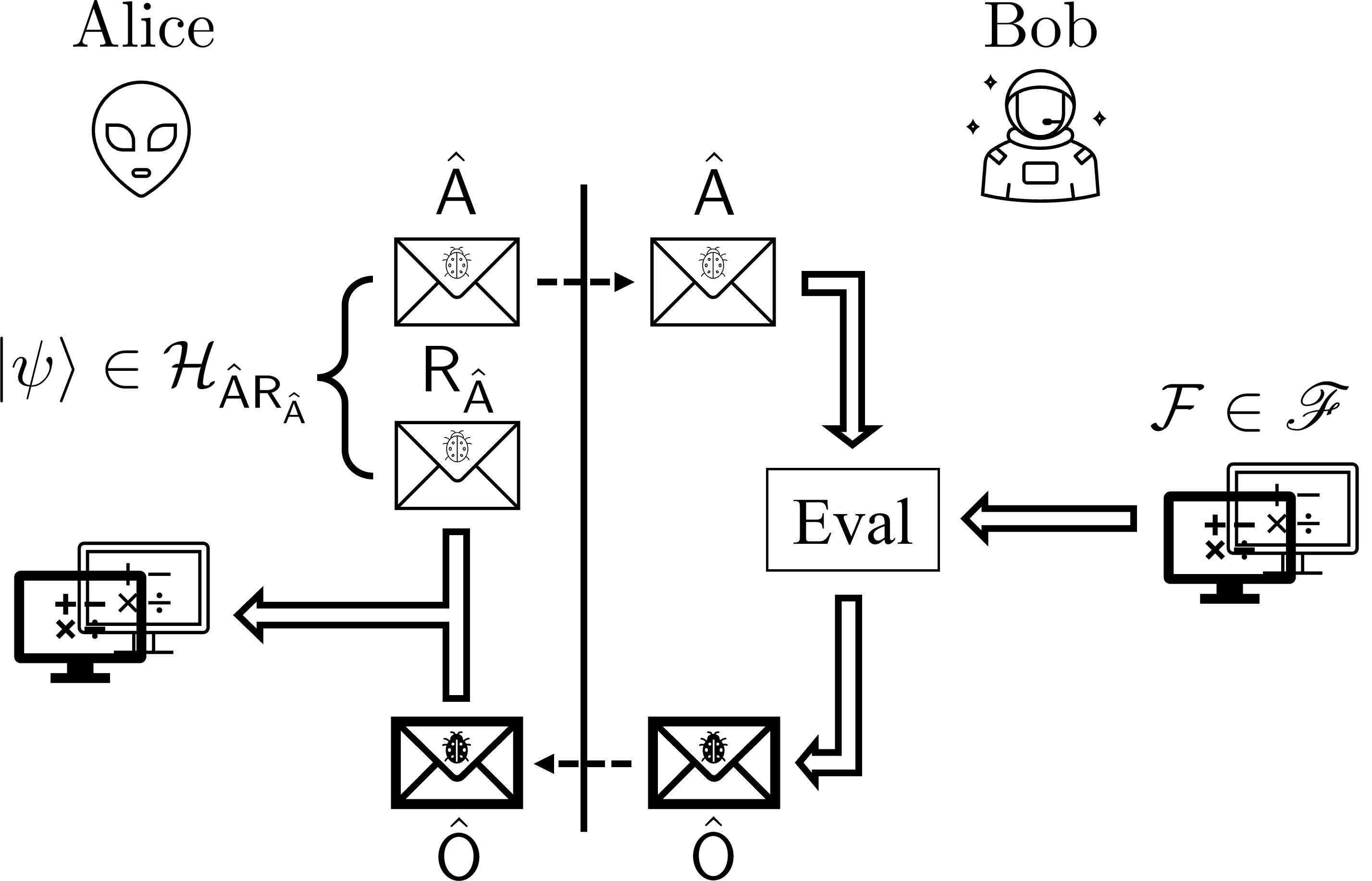}
    \caption{The actual protocol with a malicious Alice and an honest Bob}
    \label{fig:qhe_cheating_alice}
\end{figure}

In the ideal protocol, Alice sends Alice's input $\rho'\in\mathscr{S}(\A)$ (whose purification is $\ket{\psi'}\in{\A \R_{\A}}$) and Bob sends Bob's channel $\mathcal{F}\in \mathscr{F}$ to Charlie. Charlie applies $\mathcal{F}$ on $\rho'$ according to Bob's channel $\mathcal{F}$ and sends Charlie's message $\mathcal{F}(\rho')\in\mathscr{S}(\O)$ to Alice. Alice finally possesses the joint state $\mathcal{F}(\proj{\psi'})\in \mathscr{S}({\O \R_{\A}})$ of Charlie's message and Alice's referencing system. Alice further applies a post-processing channel $\mathcal{N}\in {\rm CPTP}({\O \R_{\A}},{\hat{\O} \R_{\hat{\A}}})$ on $\mathcal{F}(\proj{\psi'})$ and obtains Alice's output $\mathcal{N}(\mathcal{F}(\proj{\psi'}))\in \mathscr{S}({\hat{\O} \R_{\hat{\A}}})$. 

\begin{figure}[!htpb]
    \centering
    \includegraphics[scale=0.35]{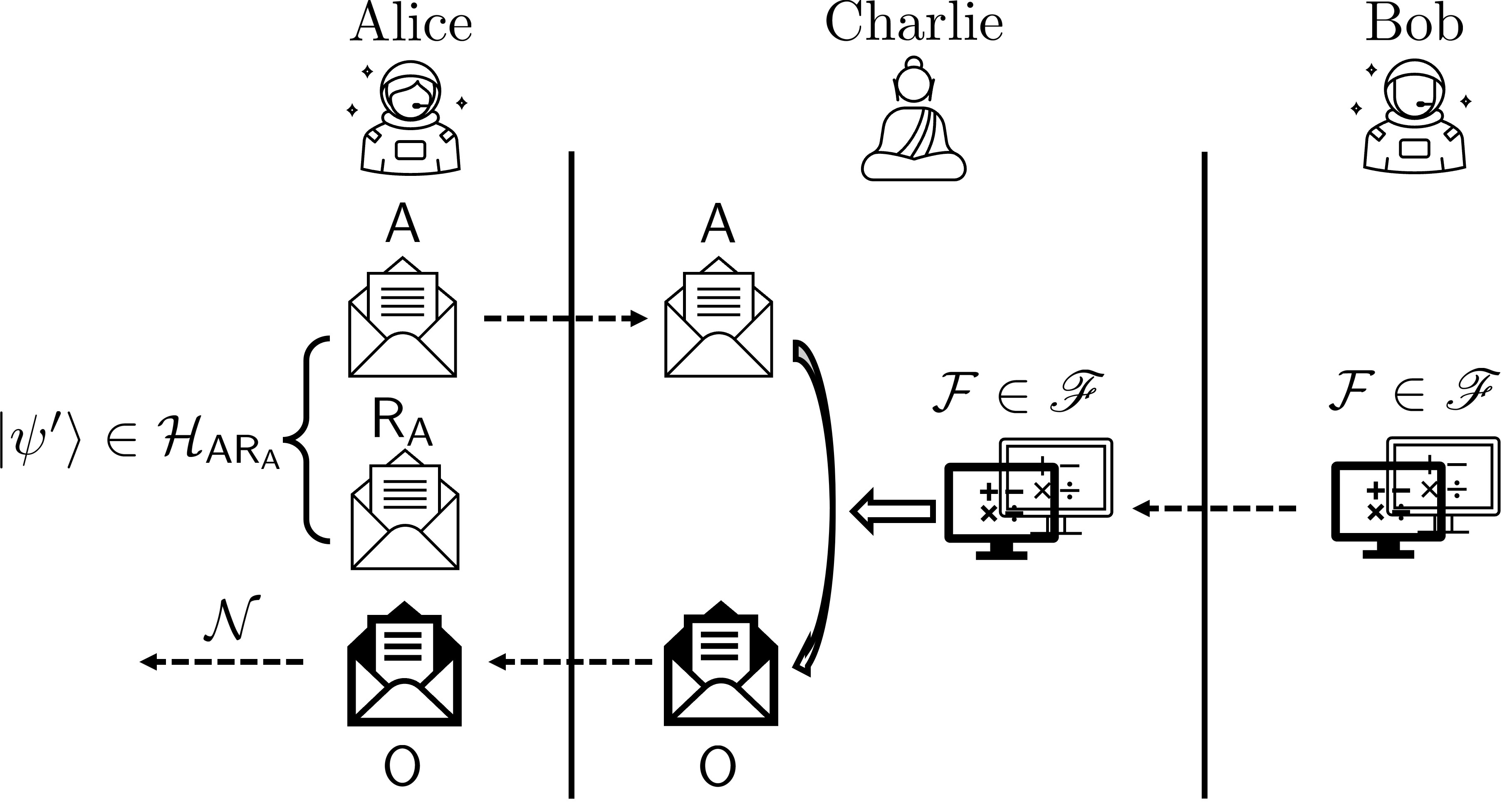}
    \caption{The ideal protocol with a trusted Charlie }
    \label{fig:qhe_ideal}
\end{figure}

\begin{definition}[Information-theoretic circuit privacy of quantum homomorphic encryption]\label{def:circuit_privacy}
    A quantum homomorphic encryption protocol $(\mathscr{F}, \KeyGen,\Enc, \Eval,\Dec)$ is $\epsilon_c$-circuit private if for any purification of Alice's message $\ket{\psi}\in {\hat{\A} \R_{\hat{\A}}}$ in the actual protocol, there exists a purification of Alice's input $\ket{\psi'}\in {\A \R_{\A}}$ and a post-processing channel $\mathcal{N}\in {\rm CPTP}({\O \R_{\A}},{\hat{\O} \R_{\hat{\A}}})$ in the ideal protocol such that for any $\mathcal{F}\in \mathscr{F}$
    \begin{align}\label{eqn:circuit_privacy}
        \Delta\left((\hat{\mathcal{F}}\otimes \mathcal{I}_{\R_{\hat{\A}}})[\proj{\psi}_{\hat{\A} \R_{\hat{\A}}}],\mathcal{N}\left((\mathcal{F}\otimes \mathcal{I}_{\R_{\A}})(\proj{\psi'}_{\A \R_{\A}})\right) \right)\leq \epsilon_c, 
    \end{align}
    Specifically when $\epsilon_c=0$, we call that the protocol is perfectly circuit private. 
\end{definition}
Our definition of circuit privacy is a good definition for the trivial protocol where Alice sends plaintexts to Bob. The trivial protocol where Alice sends plaintexts to Bob has no data privacy and perfect circuit privacy. We can write circuit privacy in terms of an optimization problem which is useful both theoretically and numerically. 
\begin{remark}\label{rmk:circuit_privacy}
    The circuit privacy can be written in the form of an optimization problem. 
    \begin{align}
        \epsilon_c = \max_{\ket{\psi}\in {\hat{\A}\R_{\hat{\A}}}} \min_{\substack{\ket{\psi'}\in{\A\R_{\A}} \\ \mathcal{N}\in {\rm CPTP}(\O\R_{\hat{\A}},\hat{\O}\R_{\hat{\A}})}} \max_{\mathcal{F} \in \mathscr F} \Delta\left((\hat{\mathcal{F}}\otimes \mathcal{I}_{\R_{\hat{\A}}})[\proj{\psi}_{\hat{\A}\R_{\hat{\A}}}],\mathcal{N}((\mathcal{F}\otimes \mathcal{I}_{\R_{\A}})(\proj{\psi'}_{\A\R_{\A}}))\right). 
    \end{align}
\end{remark}

\subsection{Quantum oblivious transfer}\label{subsec:definitions_of_qot}

Quantum oblivious transfer is a vital quantum cryptographic primitive. Following~\cite{Chailloux_2013,Chailloux_2016,Amiri_2021}, we consider two types of quantum oblivious transfer protocols,
namely standard oblivious transfer and semi-random oblivious transfer in our work which we define in Definition~\ref{def:standard_ot} and Definition~\ref{def:semi_random_ot}, respectively.
In Theorem~\ref{thm:equivalence_between_ot_and_srot}, we prove that standard oblivious transfer and semi-random oblivious transfer are equivalent by constructing reductions between them in both directions.

In standard oblivious transfer, Bob possesses two data bits, and Alice interacts with Bob to learn one specific data bit, as shown in Figure~\ref{fig:standard_ot}. Furthermore, Alice does not want Bob to know which data bit Alice desires. Bob does not want Alice to know both data bits. Next, we define standard oblivious transfer with $\delta$-completeness and soundness against a cheating Alice or Bob following Ref.~\cite[Definition~6]{Chailloux_2013}. 

\begin{figure}[!htpb]
	\centering
	\includegraphics[scale=0.35]{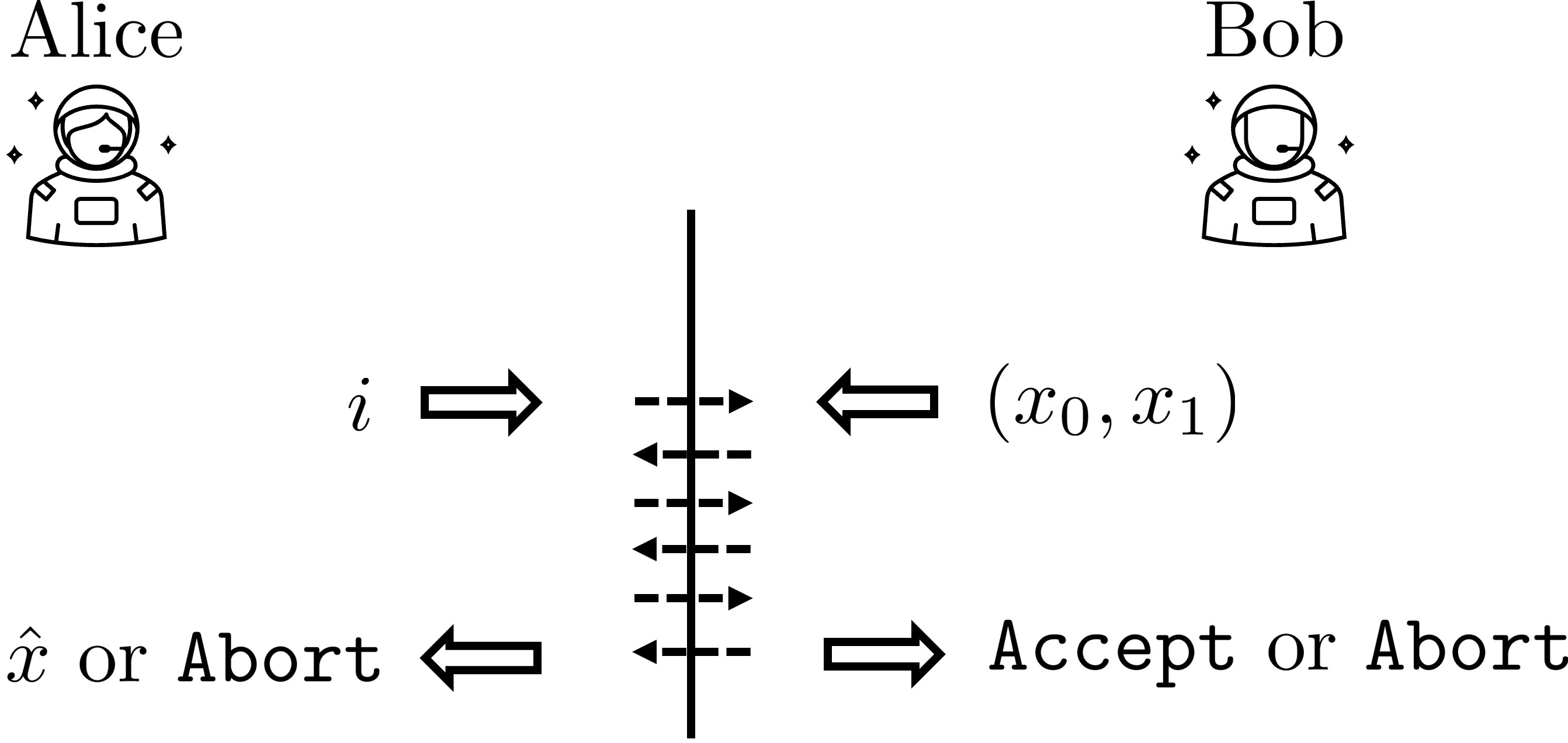}
	\caption{Standard oblivious transfer. Alice inputs $i$ and Bob inputs $(x_0,x_1)$. Then Alice and Bob perform rounds of communication. Alice outputs $\hat{x}$ or $\Abort$ and Bob outputs $\Accept$ or $\Abort$.  }
	\label{fig:standard_ot}
\end{figure}

\begin{definition}[Standard oblivious transfer]\label{def:standard_ot}
    A standard oblivious transfer protocol with $\delta$-completeness, $P_A^\star$-soundness against a cheating Alice and $P_B^\star$-soundness against a cheating Bob is a two-party protocol where Alice begins with $i\in \{0,1\}$ and Bob begins with $(x_0,x_1)\in \{0,1\}^2$, and Alice ends with ${\rm output}_A \in\{0,1\}\cup\{\Abort\}$, and Bob ends with ${\rm output}_B\in\{{\Accept},{\Abort}\}$. If Alice does not output $\Abort$, we say that Alice accepts and outputs $\hat{x}\in\{0,1\}$.
    \begin{align}
        \begin{array}{cccc}
            {\rm StandardOT}:& {\rm Alice} \times {\rm Bob} & \rightarrow & {\rm Alice} \times {\rm Bob}  \\
            & \{0,1\} \times \{0,1\}^2 & \rightarrow & \{0,1\}\cup\{\Abort\} \times \{{\Accept},{\Abort}\}  \\
            & i\times (x_0,x_1)  & \mapsto & {\rm output}_A\times {\rm output}_B  \\
        \end{array}. 
    \end{align}
    The following properties should be satisfied
    \begin{itemize}
        \item Completeness: If Alice and Bob are both honest, then both parties accept, and with a probability of at least $1-\delta$, $\hat{x}=x_{i}$. That is, when both are honest, 
        \begin{align}
            \Pr[({\rm Both\ accept})] =1,
        \end{align}
        and
        \begin{align}
            \Pr[(\hat{x}=x_{i})] \geq 1-\delta. 
        \end{align}
        Note that the above equations should hold for any choice $(x_0,x_1)$. 
        \item Soundness against a cheating Alice: Suppose that Bob's $(x_0,x_1)$ is uniformly random. With a probability of at most $P_{A}^{\star}$, a cheating Alice can guess $(\hat{x}_0,\hat{x}_1)$ for an honest Bob's $(x_0,x_1)$ correctly and Bob accepts. That is, when only Bob is honest, 
        \begin{align}
            \Pr[((\hat{x}_0,\hat{x}_1)=(x_0,x_1)) \land ({\rm Bob\ accepts}) ] \leq P_A^\star. 
        \end{align}
        \item Soundness against a cheating Bob: Suppose that Alice's $i$ is uniformly random. With a probability of at most $P_{B}^{\star}$, a cheating Bob can guess $\hat{i}$ for an honest Alice's $i$ correctly and Alice accepts. That is, when only Alice is honest, 
        \begin{align}
            \Pr[(\hat{i}=i)\land({\rm Alice\ accepts})] \leq P_B^\star. 
        \end{align}
    \end{itemize}
\end{definition}

In semi-random oblivious transfer, Bob possesses two data bits, and Alice interacts with Bob to learn one data bit uniformly at random and the index of the data bit, as is shown in Figure~\ref{fig:semi_random_ot}. Furthermore, Alice does not want Bob to know the index of the data bit. Bob does not want Alice to know both data bits. We formally define semi-random oblivious transfer with $\delta$-completeness by extending~\cite[Definition~2]{Amiri_2021}. 

\begin{figure}[!htpb]
	\centering
	\includegraphics[scale=0.35]{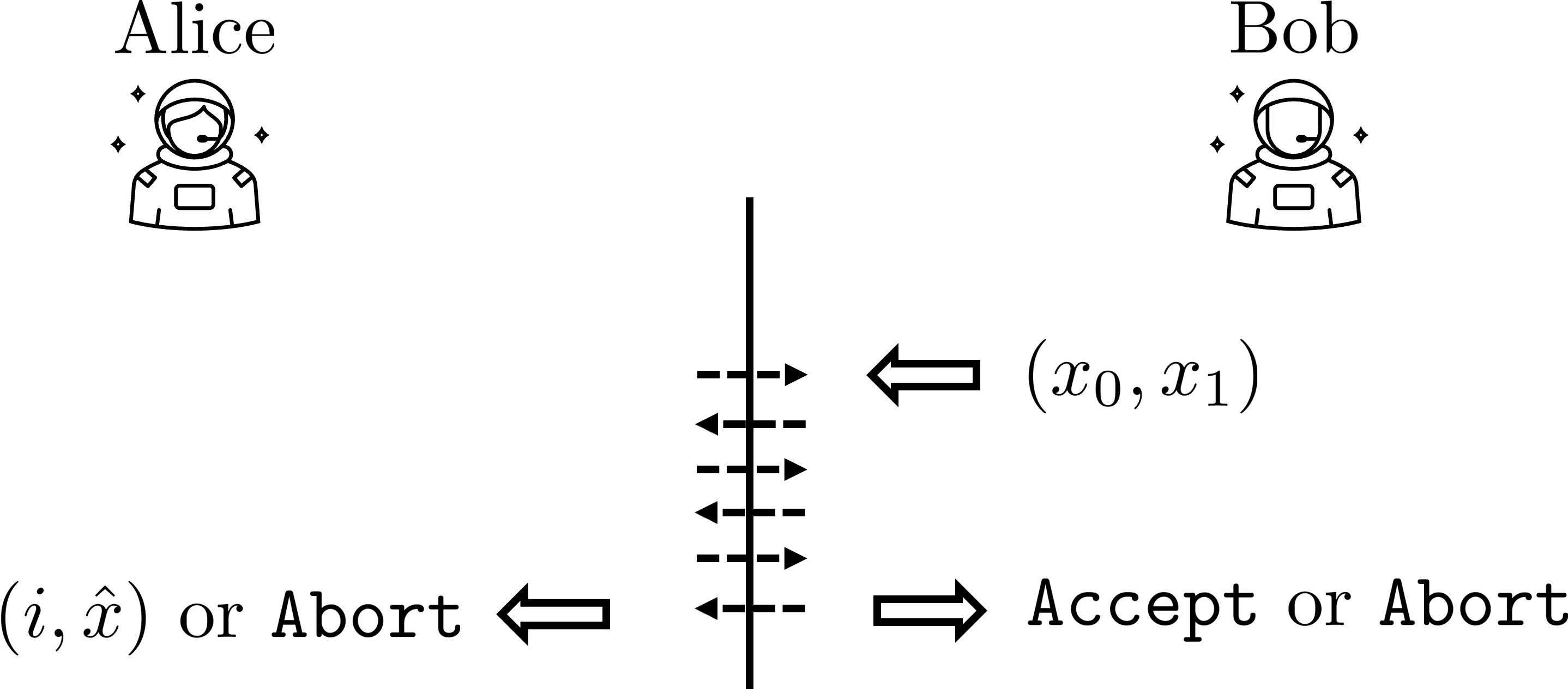}
	\caption{Semi-random oblivious transfer. Bob inputs $(x_0,x_1)$. Then Alice and Bob perform rounds of communication. Alice outputs $(i,\hat{x})$ or $\Abort$ and Bob outputs $\Accept$ or $\Abort$.  }
	\label{fig:semi_random_ot}
\end{figure}

\begin{definition}[Semi-random oblivious transfer]\label{def:semi_random_ot}
    A semi-random oblivious transfer protocol is a two-party protocol where Bob begins with $(x_0,x_1)\in\{0,1\}^2$, Alice ends with ${\rm output}_A\in \{0,1\}^2\cup\{\Abort\}$, and Bob ends with $ {\rm output}_B\in \{{\Accept},{\Abort}\}$. If Alice does not output $\Abort$, we say that Alice accepts and outputs $(i,\hat{x})\in\{0,1\}^2$.
    \begin{align}
        \begin{array}{cccc}
            {\rm SemirandomOT}:& {\rm Bob} & \rightarrow & {\rm Alice}\times {\rm Bob}  \\
            & \{0,1\}^2 & \rightarrow & \{0,1\}^2\cup\{\Abort\} \times \{{\Accept},{\Abort}\} \\
            & (x_0,x_1) & \mapsto & {\rm output}_A \times {\rm output}_B
        \end{array}, 
    \end{align}
    The protocol has $\delta$-completeness, $P_A^\star$-soundness against a cheating Alice and $P_B^\star$-soundness against a cheating Bob if the following holds.
    \begin{itemize}
        \item Completeness: If Alice and Bob are both honest, then both parties accept, $i$ is uniformly random and with a probability of at least $1-\delta$, $\hat{x}=x_{i}$. That is, when both parties are honest,  
        \begin{align}
            \Pr[({\rm Both\ accept})] =1,
        \end{align}
        \begin{align}
            \Pr[(i=0)]=\Pr[(i=1)]=\frac{1}{2},  
        \end{align}
        and 
        \begin{align}
            \Pr[(\hat{x}=x_{i})] \geq 1-\delta. 
        \end{align}
        Note again that above equations should hold for any $(x_0,x_1)$. 
        \item Soundness against a cheating Alice: Suppose that Bob's $(x_0,x_1)$ is uniformly random. With a probability of at most $P_{A}^{\star}$, a cheating Alice can guess $(\hat{x}_0,\hat{x}_1)$ for an honest Bob's $(x_0,x_1)$ correctly and Bob accepts. That is, when only Bob is honest, 
        \begin{align}
            \Pr[((\hat{x}_0,\hat{x}_1)=(x_0,x_1)) \land ({\rm Bob\ accepts})] \leq P_A^\star. 
        \end{align}
        \item Soundness against a cheating Bob: With a probability of at most $P_{B}^{\star}$, a cheating Bob can guess $\hat{i}$ for an honest Alice's $i$ correctly and Alice accepts. That is, when only Alice is honest, 
        \begin{align}
            \Pr[(\hat{i}=i)\land({\rm Alice\ accepts})] \leq P_B^\star. 
        \end{align}
    \end{itemize}
\end{definition}
Following the same technique in~\cite[Proposition~9,10]{Chailloux_2013} and~\cite[Proposition~1]{Amiri_2021}, we prove that standard oblivious transfer is equivalent to semi-random oblivious transfer in Theorem~\ref{thm:equivalence_between_ot_and_srot}.
\begin{theorem}\label{thm:equivalence_between_ot_and_srot}
    A standard oblivious transfer protocol with $\delta$-completeness, $P_A^\star$-soundness against a cheating Alice and $P_B^\star$-soundness against a cheating Bob is equivalent to a semi-random oblivious transfer protocol with $\delta$-completeness, $P_A^\star$-soundness against a cheating Alice and $P_B^\star$-soundness against a cheating Bob.
\end{theorem}
 The proof works by constructing reductions between semi-random and standard oblivious transfer. 
 We can easily reduce semi-random oblivious transfer to standard oblivious transfer. A semi-random oblivious transfer protocol can be viewed as Alice chooses the index of the data bit Alice wants to learn and performs a standard oblivious transfer protocol with Bob. It is trickier to reduce standard oblivious transfer to semi-random oblivious transfer. A semi-random oblivious transfer protocol can work as a subprotocol to generate keys in a standard oblivious transfer protocol. Initially, Bob holds two keys. After a semi-random oblivious transfer protocol, Alice learns one key uniformly at random and the key index, while Bob does not know the index of the key. Alice can then encrypt Alice's request with the key index, and Bob can encrypt Bob's two data bits with two keys. In this way, Alice can only decrypt one data bit, and Bob cannot decrypt Alice's request. We provide the detailed proof in Appendix~\ref{sec:equivalence_between_ot_and_srot}.

\section{Bounds for quantum oblivious transfer}\label{sec:bounds_for_qot}

In this section, we extend the bound for quantum oblivious transfer in~\cite[Eq.~(27)]{Amiri_2021} to $\delta$-correctness with the same technique. We first describe the general scheme of semi-random oblivious transfer and then bound Alice's and Bob's cheating probabilities by proposing their cheating strategies. 

The general scheme of semi-random oblivious transfer with $N$ rounds of communication~\cite[Section~IIIA]{Amiri_2021} is illustrated in Figure~\ref{fig:qot_scheme}. Alice and Bob keep their memories private and exchange their messages publicly.
Each time either Alice or Bob receives the message, they apply a unitary that acts jointly on the memory and the message. In the last step, they measure their state to obtain their output. Any semi-random oblivious transfer can be described by such a general scheme.

\begin{figure}[!htpb]
	\centering
	\includegraphics[scale=0.35]{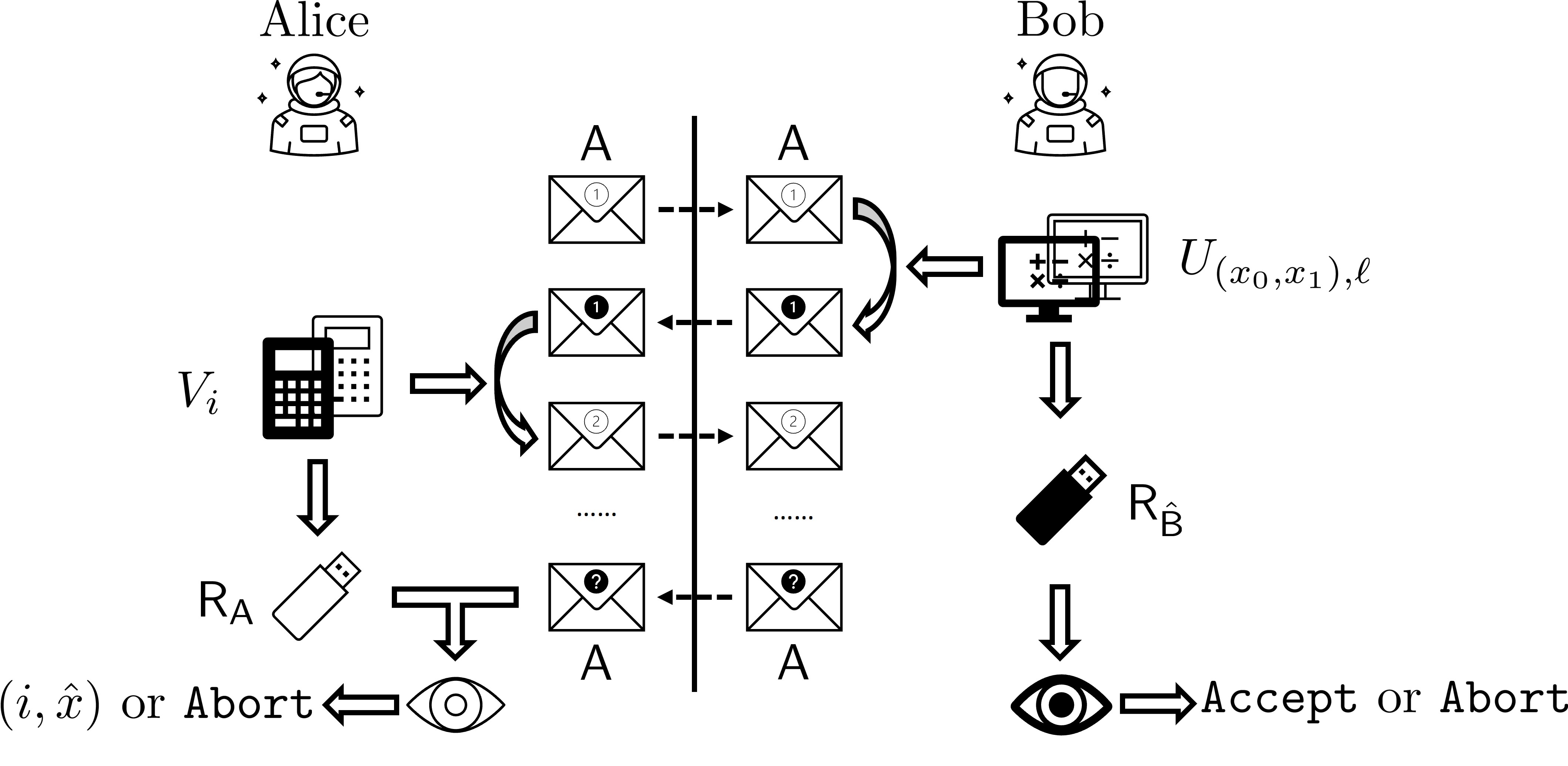}
	\caption{The scheme of quantum oblivious transfer. }
	\label{fig:qot_scheme}
\end{figure}

Protocol~\ref{prot:srot_general_scheme} formally describes the general scheme depicted in Figure~\ref{fig:qot_scheme}. 
Alice begins with the state $\ket{\psi}\in {\R_{\A} \A}$, where $\R_{\A}$ is Alice's referencing system and $\A$ is Alice's message. Bob begins with the state $\ket{0} \in {\R_{\B}}$ and the data bits $(x_0,x_1)\in\{0,1\}^2$, where $\R_{\B}$ is Bob's referencing system. The joint state of Alice and Bob is $\ket{\psi}\ket{0}\ket{x_0,x_1}\in{\A \R_{\A}\B \R_{\B}}$. Bob prepares the program $\ket{x_0,x_1}\in{\B}$. Then Alice and Bob repeat $N$ rounds of communication: Alice sends $\A$ to Bob; Bob applies a unitary $U_{(x_0,x_1),{\ell}}\in {\rm U}(\A\R_{\B})$ that acts trivially on $\R_{\A}$ and depends on $(x_0,x_1)$; Bob sends $\A$ back to Alice; Alice applies a unitary $ V_{{\ell}}\in {\rm U}(\A \R_{\A})$ that acts trivially on $\R_{\B}$. After $N$ iterations, Alice and Bob's joint state is $\ket{\psi^{(x_0,x_1)}}= V_{N}U_{(x_0,x_1),N}...V_{1} U_{(x_0,x_1),1}\ket{\psi}\in \mathscr{S}({\A \R_{\A}\R_{\B}})$. Alice's final state is $\sigma_{(x_0,x_1)}=\Tr_{\R_{\B}}(\proj{\psi^{(x_0,x_1)}})\in \mathscr{S}({\A \R_{\A}})$. We denote the maximum fidelity between distinct pairs of Alice's final states by $f$, i.e. 
\begin{align}\label{eqn:maxinum_fidelity}
    f =\max\left\{F(\sigma_{(x_0,x_1)}, \sigma_{(x_0',x_1')}),(x_0,x_1)\neq (x_0',x_1')\right\}.
\end{align}
Alice measures a projective measurement $\{\Pi_{\rm Alice}^{\Accept},\Pi_{\rm Alice}^{\Abort}\}$ on $\A\R_{\A}$ to determine if Alice accepts, where $\Pi_{\rm Alice}^{\Accept}$ is the projector onto Alice's accepting subspace while $\Pi_{\rm Alice}^{\Abort}$ is the projector onto Alice aborting subspace. Similarly, Bob measures $\{\Pi_{\rm Bob}^{\Accept},\Pi_{\rm Bob}^{\Abort}\}$ on $\R_{\B}$ to determine if Bob accepts, where $\Pi_{\rm Bob}^{\Accept}$ is the projector onto Bob's accepting subspace while $\Pi_{\rm Bob}^{\Abort}$ is the projector onto Bob aborting subspace. If both accept, Alice performs a positive operator valued measure $\{N_{(i,\hat{x})}=E_{(i,\hat{x})}^\dagger E_{(i,\hat{x})}\}_{i,\hat{x}\in\{0,1\}}$ on $\R_{\A}$, where $E_{(i,\hat{x})}$ relates the post-measurement state to the pre-measurement state. Alice's output is the measurement outcome $(i,\hat{x})$.

\begin{algorithm}
    \caption{Scheme of semi-random oblivious transfer with $N$ rounds of communication}
    \label{prot:srot_general_scheme}
    \begin{algorithmic}[1]
        \Require $(x_0,x_1)\in\{0,1\}^2$. \Comment{Bob's input}
        \Ensure ${\rm output}_{A}\in\{0,1\}^2\cup\{\Abort\}$, ${\rm output}_B\in\{\Accept,\Abort\}$. \Comment{Alice's and Bob's outputs}
        \State Alice: Prepare $\ket{\psi}\in {\R_{\A} \A}$ 
        \State Bob: Prepare $\ket{0}\in {\R_{\B}}$.  
        \For{$\ell = 1$ to $\ell = N$} \Comment{Alice and Bob performs $N$ rounds of communication. }
            \State Alice: Send $\A$ to Bob. 
            \State Bob: Perform $U_{(x_0,x_1),{\ell}}\in {\rm U}(\A\R_{\B})$. 
            \State Bob: Send $\A$ back to Alice. 
            \State Alice: Perform $ V_{{\ell}}\in {\rm U}(\A \R_{\A})$. 
        \EndFor
        \State Alice: Measure $\{\Pi_{\rm Alice}^{\Accept},\Pi_{\rm Alice}^{\Abort}\}$ on $\A \R_{\A}$ with outcome $A$. \Comment{Alice measures to determine whether to accept. }
        \State Bob: Measure $\{\Pi_{\rm Bob}^{\Accept},\Pi_{\rm Bob}^{\Abort}\}$ on $ \R_{\B}$ with outcome $B$. \Comment{Bob measures to determine whether to accept. }
        \State Bob: ${\rm output}_B\leftarrow B$. 
        \If{$A=\Abort$}
            \State Alice: ${\rm output}_A\leftarrow {\Abort}$. 
        \Else
            \State Alice: Measure $\{N_{(i,\hat{x})}\}_{i,\hat{x}\in\{0,1\}}$ on $\A\R_{\A}$ with outcome $(i,\hat{x})$. \Comment{Alice measures to determine $(i,\hat{x})$. }
            \State Alice: ${\rm output}_A\leftarrow (i,\hat{x})$. 
        \EndIf
    \end{algorithmic}
\end{algorithm}

Applying a similar method as in~\cite[Section IIIB,C,D]{Amiri_2021}, we can obtain a bound for the semi-random oblivious transfer protocol. 
\begin{theorem}\label{thm:srot_bound}
    Any semi-random oblivious transfer protocol and any standard oblivious transfer protocol with $\delta$-completeness, $P_A^\star$-soundness against a cheating Alice and $P_B^\star$-soundness against a cheating Bob satisfies 
    \begin{align}\label{eqn:srot_bound}
        P_A^\star+2P_B^\star + 4\sqrt{\delta}\geq 2.
    \end{align}
\end{theorem}

In order to show a violation of soundness, it suffices to exhibit a specific cheating strategy. Therefore, we will only deal with certain strategies of cheating Alice and Bob in the proof. Alice's cheating strategy involves performing a pretty-good measurement in the last step to try to learn Bob's input. Bob can input a well-chosen superposition at the beginning and measure at the end to try to learn Alice's input.

\begin{proof}
    Here we prove a bound for semi-random oblivious transfer, and due to the equivalence between semi-random oblivious transfer and standard oblivious transfer in Theorem~\ref{thm:equivalence_between_ot_and_srot}, the same bound applies to standard oblivious transfer. Consider a semi-random oblivious transfer protocol with $\delta$-completeness described by Protocol~\ref{prot:srot_general_scheme}. The positive operator valued measure $\{N_{(i,\hat{x})}\}_{i,\hat{x}\in\{0,1\}}$ must satisfy 
    \begin{align}\label{eqn:required_measurement}
        \Pr[({\rm Alice\ output\ }(i,\hat{x}) )|{\rm Bob\ input\ }(x_0,x_1)]=\Tr(N_{(i,\hat{x})}\sigma_{(x_0,x_1)})=\left\{
        \begin{array}{cc}
            \frac{1- \theta_{i,(x_0,x_1)}}{2} & \hat{x}=x_{i}\\
            \frac{\theta_{i,(x_0,x_1)}}{2} & \hat{x}\neq x_{i}
        \end{array}
        \right. 
        .
    \end{align}
    where $\theta_{i,(x_0,x_1)}$ is the error probability that depends on $i$ and $(x_0,x_1)$. In order to satisfy the $\delta$-completeness of standard oblivious transfer, $\theta_{i,(x_0,x_1)}$ must satisfy for any $(x_0,x_1)$
    \begin{align}
        0\leq \sum_{i\in\{0,1\}}\theta_{i,(x_0,x_1)} \leq 2\delta. 
    \end{align}
    
    We first discuss the soundness against Alice. Suppose that Alice is malicious and Bob is honest. Alice can follow Protocol~\ref{prot:srot_general_scheme} through Step 14 while performing the pretty good measurement $\{N_{(\hat{x}_0,\hat{x}_1)}^{\rm PGM}=S^{-\frac{1}{2}}\sigma_{(\hat{x}_0,\hat{x}_1)}S^{-\frac{1}{2}}\}_{\hat{x}_0,\hat{x}_1\in\{0,1\}}$ where $S=\sum_{x_0,x_1\in \{0,1\}}\sigma_{(x_0,x_1)}$ instead of the required measurement $\{N_{{(i,\hat{x})}}\}_{i,\hat{x}\in\{0,1\}}$ in Step 15 and guess $(\hat{x}_0,\hat{x}_1)$ for $(x_0,x_1)$. A cheating Alice will not be caught by an honest Bob since Alice follows the protocol until Bob accepts. If $(x_0,x_1)$ is uniformly random, then the probability that Alice can guess correctly is bounded in~\cite[Theorem~3.1]{Audenaert_2014},
    and satisfies the inequality
    \begin{align}\label{eqn:correct_probability_PGM}
        \Pr[{(\hat{x}_0,\hat{x}_1)}=(x_0,x_1)] \geq 1- \frac{1}{8}\sum_{\substack{x_0,x_1,x_0',x_1'\in \{0,1\}\\ (x_0,x_1)\neq (x_0',x_1')}} F(\sigma_{(x_0,x_1)},\sigma_{(x_0',x_1')}). 
    \end{align}
    As shown in Eq.~\eqref{eqn:required_measurement}, the required measurement $\{N_{(i,\hat{x})}\}_{i,\hat{x}\in\{0,1\}}$ distinguishes $\sigma_{(x_0,x_1)}$ and $\sigma_{(\overline{x}_0,\overline{x}_1)}$ with a probability of at least $1-\delta$. Therefore, the Holevo-Helstrom theorem~\cite{Helstrom_1967,Holevo_1973} (see~\cite[Theorem 3.4]{Watrous_2018} for a modern version) implies that
    \begin{align}
        \frac{1}{2}(1+\Delta(\sigma_{(x_0,x_1)},\sigma_{(\overline{x}_0,\overline{x}_1)}))\geq 1-\delta. 
    \end{align}
    Applying the Fuchs-van de Graaf inequality in~\cite[Theorem 1]{Fuches_1999} 
    \begin{align}
        \Delta(\sigma_{(x_0,x_1)},\sigma_{(\overline{x}_0,\overline{x}_1)}) \leq \sqrt{1-F(\sigma_{(x_0,x_1)},\sigma_{(\overline{x}_0,\overline{x}_1)})^2}, 
    \end{align}
    we obtain
    \begin{align}\label{eqn:fidelity_complement}
        F(\sigma_{(x_0,x_1)},\sigma_{(\overline{x}_0,\overline{x}_1)})\leq 2\sqrt{\delta(1-\delta)}. 
    \end{align}
    Recall the maximum fidelity $f$ defined in \eqref{eqn:maxinum_fidelity}. By reordering terms in Eq.~\eqref{eqn:correct_probability_PGM}, we obtain
    \begin{align}
        \Pr[{(\hat{x}_0,\hat{x}_1)}=(x_0,x_1)]\geq 1 & - \frac{1}{8}\sum_{\substack{x_0,x_1,x_0',x_1'\in\{0,1\}\\(\overline{x}_0,\overline{x}_1)=(x_0',x_1')}} F(\sigma_{(x_0,x_1)},\sigma_{(x_0',x_1')}) \nonumber\\ 
        & - \frac{1}{8}\sum_{\substack{x_0,x_1,x_0',x_1'\in\{0,1\}\\
        (x_0,x_1)\neq (x_0',x_1')\\(\overline{x}_0,\overline{x}_1)\neq(x_0',x_1') }} F(\sigma_{(x_0,x_1)},\sigma_{(x_0',x_1')}). 
    \end{align}
    Applying Eq.~\eqref{eqn:fidelity_complement}, we obtain 
    \begin{align}
        \Pr[(\hat{x}_0,\hat{x}_1)=(x_0,x_1)]\geq 1 -f-\sqrt{\delta(1-\delta)}. 
    \end{align}
    Therefore, a cheating Alice can guess an honest Bob's $(x_0,x_1)$ and both parties accept with a probability of at least  $1-f-\sqrt{\delta(1-\delta)}$, and thus 
    \begin{align}\label{eqn:soundness_against_alice}
        P_A^\star \geq 1- f - \sqrt{\delta(1-\delta)}. 
    \end{align}
    
    Second we discuss the soundness against Bob. Suppose that Alice is honest and Bob is malicious. Let $(x_0,x_1)$ and $(x_0',x_1')$ be the two data bit strings corresponding to the maximum fidelity $f$ defined in \eqref{eqn:maxinum_fidelity}, i.e. $f=F(\sigma_{(x_0,x_1)}, \sigma_{(x_0',x_1')})$. There is at least one different bit between $(x_0,x_1)$ and $(x_0',x_1')$. Let us assume that $x_1=0$ and $x_1'=1$, and other cases follow analogously. According to Uhlmann's theorem~\cite{Uhlmann_1976} (or see a textbook e.g.~\cite[Theorem 3.22]{Watrous_2018}), we can find a purification $\ket{\phi^{(x_0,0)}}\in {\A \R_{\A}\R_{\B}}$ of $\sigma_{(x_0,0)}$ and a purification $\ket{\phi^{(x_0',1)}}\in {\A \R_{\A}\R_{\B}}$ of $\sigma_{(x_0',1)}$ such that $f = |\braket{\phi^{(x_0',1)}}{\phi^{(x_0,0)}}|$. Due to the unitary equivalence of purifications~\cite[Theorem~2.12]{Watrous_2018}, Bob can perform local unitaries $U_{(x_0,0)}\in{\rm U}(\R_{\B})$ and $U_{(x_0',1)}\in{\rm U}(\R_{\B})$ such that $\ket{\phi^{(x_0,0)}}=U_{(x_0,0)}\ket{\psi^{(x_0,0)}}$ and $\ket{\phi^{(x_0',1)}}=U_{(x_0',1)}\ket{\psi^{(x_0',1)}}$. 
    
    Based on the above facts, Bob can cheat by considering the two inputs $(x_0,0)$ and $(x_0',1)$, preparing the program in superposition and otherwise honestly following the protocol, and making a measurement on his state in the end to try to learn Alice's input. More precisely, instead of following Protocol~\ref{prot:srot_general_scheme}, Bob will prepare the superposition $\ket{(x_0,0)}+\ket{(x_0',1)}\in {\B}$ and apply the controlled unitary  $\sum_{x_0'',x_1''\in\{0,1\}}U_{(x_0'',x_1''),\ell}\otimes \proj{x_0'',x_1''}_{\B} \in {\rm U}(\A \B\R_{\B} )$ in Step 5 of each round of communication. Before Step 9, the joint state of both parties is the superposition $\ket{\psi^{(x_0,0)}}\ket{x_0,0}+\ket{\psi^{(x_0',1)}}\ket{x_0',1}$. A cheating Bob will not be caught by an honest Alice because $\Pi^{\Accept}_{\rm Alice}\ket{\psi^{(x_0,0)}}=\ket{\psi^{(x_0,0)}}$ and $\Pi^{\Accept}_{\rm Alice}\ket{\psi^{(x_0',1)}}=\ket{\psi^{(x_0',1)}}$, hence
    $\Pi^{\Accept}_{\rm Alice}\left(\ket{\psi^{(x_0,0)}}\ket{x_0,0}+\ket{\psi^{(x_0',1)}}\ket{x_0',1}\right)=\ket{\psi^{(x_0,0)}}\ket{x_0,0}+\ket{\psi^{(x_0',1)}}\ket{x_0',1}$. After Step 15, Bob wants to learn $i$. This can be done by a unitary $U_{(x_0,0)}\otimes \proj{x_0,0}_{\B}+ U_{(x_0',1)}\otimes \proj{x_0',1}_{\B}+ \mathbbm{I}_{\R_{\B}} \otimes (\mathbbm{I}_{\B}- \proj{x_0,0}_{\B}-\proj{x_0',1}_{\B}) \in {\rm U}( \B \R_{\B})$ followed by a measurement $\{N_{\hat{i}}^{\rm OPT}\}_{\hat{i}\in\{0,1\}}$ on $\B$. This can bound the probability that the estimate $\hat{i}$ is equal to $i$. The joint state of Alice and Bob $\ket{\Phi}\in{\A \R_{\A}\B \R_{\B}\E}$ before Alice and Bob measure in Step 15 is
    \begin{align}
        \ket{\Phi}= \frac{1}{\sqrt{2}}\left(\ket{\phi^{(x_0,0)}} \ket{x_0,0} + \ket{\phi^{(x_0',1)}} \ket{x_0',1}\right). 
    \end{align}
    The reduced density matrix $\rho_{\B}\in \mathscr{S}(\B)$ is 
    \begin{align}\label{eqn:before_measurement}
        \rho_{\B} = \frac{1}{2} \left(\proj{x_0,0}_{\B}+\proj{x_0',1}_{\B}\right)+\frac{1}{2}\left(\braket{\phi^{(x_0',1)}}{\phi^{(x_0,0)}}\ketbra{x_0,0}{x_0',1}_{\B}+ \textnormal{h.c.} \right).
    \end{align}
    where $\textnormal{h.c.}$ stands for the Hermitian conjugate of the previous term. After Alice performs the required measurement $\{N_{i,\hat{x}}\}_{i,\hat{x}\in\{0,1\}}$ on ${\A \R_{\A}}$ to obtain $(i,\hat{x})$ in Step 15, the post-measurement reduced density matrices $\rho_{\B}^{(i)}\in\mathscr{S}(\B)$ given $i$ is 
    \begin{align}
        \rho_{\B}^{(i)} = \sum_{\hat{x}\in\{0,1\}} \Tr_{\A \R_{\A} \R_{\B}}(E_{(i,\hat{x})}\proj{\Phi}E_{(i,\hat{x})}^\dagger),
    \end{align}
    or more explicitly
    \begin{align}
        \rho_{\B}^{(i)} = & \ 
        \frac{1}{2} \left(\proj{x_0,0}_{\B}+\proj{x_0',1}_{\B}\right) \nonumber\\
        & \ +\sum_{\hat{x}\in\{0,1\}}\left(\sandwich{\phi^{(x_0',1)}}{ N_{(i,\hat{x})}}{\phi^{(0,0)}}\ketbra{(x_0,0)}{(x_0',1)}_{\B}+\textnormal{h.c.}\right).
    \end{align}
    The difference between $\rho_{\B}^{(0)}$ and $\rho_{\B}^{(1)}$ is  
    \begin{align}
        & \rho_{\B}^{(0)} - \rho_{\B}^{(1)} = \nonumber \\
        & 
        \left(
        \begin{matrix}
            0 & \sum_{\hat{x}\in\{0,1\}}\sandwich{\phi^{(x_0',1)}}{ N_{(0,\hat{x})}-N_{(1,\hat{x})}}{\phi^{(x_0,0)}} \\
            \sum_{\hat{x}\in\{0,1\}}\sandwich{\phi^{(x_0,0)}}{ N_{(0,\hat{x})}-N_{(1,\hat{x})}}{\phi^{(x_0',1)}} & 0 \\
        \end{matrix}\right). 
    \end{align}
    The trace norm of a two-dimensional Hermitian matrix with only off-diagonal elements is just the absolute value of its off-diagonal elements. Hence we obtain
    \begin{align}
        \Delta(\rho_{\B}^{(0)},\rho_{\B}^{(1)}) = \left|\sum_{\hat{x}\in\{0,1\}}\sandwich{\phi^{(x_0',1)}}{ N_{(0,\hat{x})}}{\phi^{(x_0,0)}}-\sum_{\hat{x}\in\{0,1\}}\sandwich{\phi^{(x_0',1)}}{ N_{(1,\hat{x})}}{\phi^{(x_0,0)}}\right|. 
    \end{align}
    Since $\sum_{i,\hat{x}\in\{0,1\}} N_{(i,\hat{x})} = \mathbbm{I}_{\A \R_{\A}}$, we obtain 
    \begin{align}
        \Delta(\rho_{\B}^{(0)},\rho_{\B}^{(1)}) = \left|\braket{\phi^{(0,1)}}{\phi^{(0,0)}}-2\sum_{\hat{x}\in\{0,1\}}\sandwich{\phi^{(0,1)}}{ N_{(1,\hat{x})}}{\phi^{(0,0)}}\right|. 
    \end{align}
    Applying the reverse triangle inequality, we obtain
    \begin{align}\label{eqn:trace_bound}
        \Delta(\rho_{\B}^{(0)},\rho_{\B}^{(1)}) \geq \left|\braket{\phi^{(x_0',1)}}{\phi^{(x_0,0)}}\right|-2\sum_{\hat{x}\in\{0,1\}}\left|\sandwich{\phi^{(x_0',1)}}{N_{(1,\hat{x})}}{\phi^{(x_0,0)}}\right|.
    \end{align}
    Furthermore, applying the Cauchy-Schwarz inequality, we obtain
    \begin{align}
        \left|\sandwich{\phi^{(x_0',1)}}{N_{(1,\hat{x})}}{\phi^{(x_0,0)}}\right|\leq \sqrt{\left|\sandwich{\phi^{(x_0,0)}}{N_{(1,\hat{x})}}{\phi^{(x_0,0)}}\right|\cdot \left|\sandwich{\phi^{(x_0',1)}}{N_{(1,\hat{x})}}{\phi^{(x_0',1)}}\right|}\leq \frac{\sqrt{2\delta}}{2}, 
    \end{align}
    where the last inequality follows from Eq.~\eqref{eqn:required_measurement}. 
    Substituting into Eq.~\eqref{eqn:trace_bound}, we find
    \begin{align}
        \Delta(\rho_{\B}^{(0)},\rho_{\B}^{(1)}) \geq f-2\sqrt{2\delta}.
    \end{align}
    Therefore, we obtain $\Pr[\hat{i}=i] \geq \frac{1}{2}( 1+f-2\sqrt{2\delta})$. Hence, a cheating Bob can guess an honest Alice's $i$ and both parties accept with a probability of at least
    \begin{align}\label{eqn:soundness_against_bob}
        P_B^\star \geq \frac{1}{2}(1+f-2\sqrt{2\delta}). 
    \end{align}
    
    Combining Eq.~\eqref{eqn:soundness_against_alice} with Eq.~\eqref{eqn:soundness_against_bob} and eliminating $F$, we obtain a bound for semi-random oblivious transfer
    \begin{align}
        P_A^\star+2P_B^\star + \sqrt{\delta(1-\delta)} +2\sqrt{2\delta}\geq 2.
    \end{align}
    By further using 
    \begin{align}
        \sqrt{\delta(1-\delta)} +2\sqrt{2\delta}\leq 4\sqrt{\delta},
    \end{align}
    we have
    \begin{align}
        P_A^\star+2P_B^\star + 4\sqrt{\delta}\geq 2.
    \end{align}
    
\end{proof}

\section{Reduction from quantum oblivious transfer to quantum homomorphic encryption}\label{sec:reduction}

In this section, we reduce quantum oblivious transfer to quantum homomorphic encryption. We construct a set of channels which can realize quantum oblivious transfer in Definition~\ref{def:sot_channels}. We then use a quantum homomorphic encryption protocol with the set of channels as a black-box subprotocol to construct a quantum oblivious transfer protocol, as shown in Figure~\ref{fig:reduction}. Alice inputs the index of the data bit, Bob inputs two data bits to the black-box subprotocol, and the black-box subprotocol outputs the desired data bit to Alice. We complete the reduction by translating the correctness, data privacy and circuit privacy of quantum homomorphic encryption into the completeness, soundness against Bob and soundness against Alice of quantum oblivious transfer in Lemma~\ref{lemma:completeness}, Lemma~\ref{lemma:soundness_against_bob} and Lemma~\ref{lemma:soundness_against_alice}, respectively.

\begin{figure}[!htpb]
	\centering
	\includegraphics[scale=0.35]{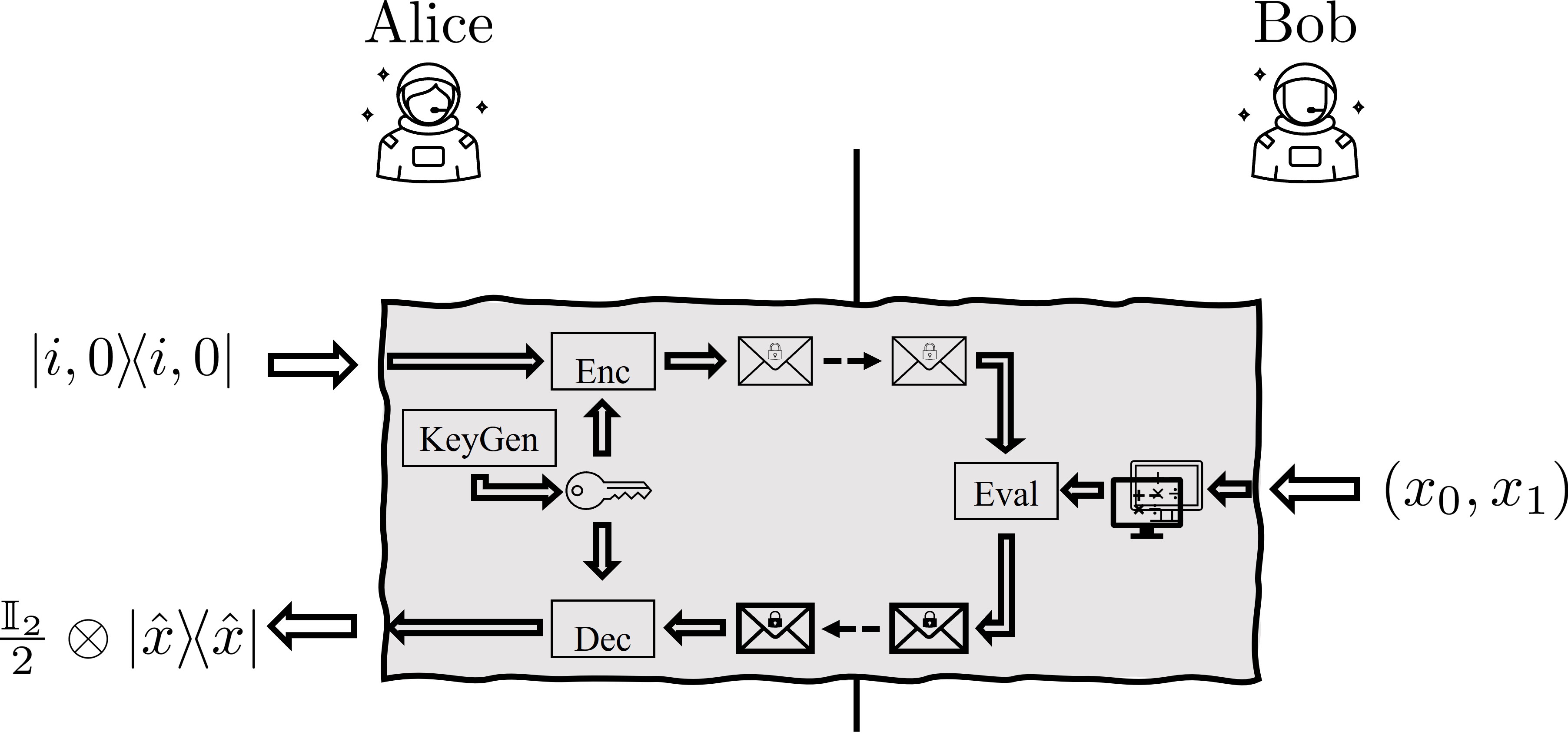}
	\caption{The reduction from quantum oblivious transfer to quantum homomorphic encryption}
	\label{fig:reduction}
\end{figure}

We now define the set of strong oblivious transfer channels. Roughly speaking, the set of strong oblivious transfer channels can realize standard oblivious transfer in a classical manner. 

\begin{definition}[Strong oblivious transfer channels]\label{def:sot_channels}
    Strong oblivious transfer channels $\mathcal{F}_{(x_0,x_1)}\in{\rm CPTP}(\A,\O)$, where $(x_0,x_1)\in\{0,1\}^2$, can be compactly given by
    \begin{align}
        \mathcal{F}_{(x_0,x_1)}(\rho)= \sum_{i,i'\in \{0,1\}}\Tr\left(\proj{i,i'}_{\A}\rho\right)\frac{\mathbbm{I}_2}{2}\otimes\proj{x_{i\oplus i'}\oplus i'}_{\O_2}.
    \end{align}
    where $\O_2$ denotes the second qubit in the output.  
\end{definition}
\begin{remark}
     Strong oblivious transfer channels $\mathcal{F}_{(x_0,x_1)}$ can be explicitly expressed as 
    \begin{align}
        \mathcal{F}_{(0,0)}(\rho) =& \frac{\mathbbm{I}_2}{2}\otimes \left( \Tr\left(\mathbbm{I}_2\otimes\proj{0}_{\A_2}  \rho\right)\proj{0}_{\O_2}+\Tr\left(\mathbbm{I}_2\otimes\proj{1}_{\A_2}  \rho\right)\proj{1}_{\O_2}\right), \\
        \mathcal{F}_{(0,1)}(\rho) =& \frac{\mathbbm{I}_2}{2}\otimes \left(\Tr\left(\proj{0}_{\A_1}\otimes \mathbbm{I}_2 \rho\right)\proj{0}_{\O_2}+ \Tr\left(\proj{1}_{\A_1}\otimes \mathbbm{I}_2 \rho\right)\proj{1}_{\O_2}\right),\\
        \mathcal{F}_{(1,0)}(\rho) =&\frac{\mathbbm{I}_2}{2}\otimes \left( \Tr\left(\proj{0}_{\A_1}\otimes\mathbbm{I}_2  \rho\right)\proj{1}_{\O_2}+ \Tr\left(\proj{1}_{\A_1} \otimes\mathbbm{I}_2 \rho\right)\proj{0}_{\O_2}\right),\\
        \mathcal{F}_{(1,1)}(\rho) =&\frac{\mathbbm{I}_2}{2}\otimes \left( \Tr\left(\mathbbm{I}_2\otimes\proj{0}_{\A_2}  \rho\right)\proj{1}_{\O_2}+\Tr\left(\mathbbm{I}_2\otimes\proj{1}_{\A_2}  \rho\right)\proj{0}_{\O_2}\right). 
    \end{align}
    where $\A_1$ and $\A_2$ denotes the first and second qubit in the input, respectively. 
\end{remark}
\begin{remark}
    When $\rho = \proj{i,0}_{\A}$, then $\mathcal{F}_{(x_0,x_1)}$ satisfies 
    \begin{align}
        \mathcal{F}_{(x_0,x_1)}(\rho) =\frac{\mathbbm{I}_2}{2}\otimes  \proj{x_i}_{\O_2}
    \end{align}
    or
    \begin{align}
        \mathcal{F}_{(0,0)}(\rho) =& \frac{\mathbbm{I}_2}{2}\otimes \proj{0}_{\O_2}, \\
        \mathcal{F}_{(0,1)}(\rho) =& \frac{\mathbbm{I}_2}{2}\otimes \proj{i}_{\O_2},\\
        \mathcal{F}_{(1,0)}(\rho) =& \frac{\mathbbm{I}_2}{2}\otimes \proj{\overline{i}}_{\O_2},\\
        \mathcal{F}_{(1,1)}(\rho) =& \frac{\mathbbm{I}_2}{2}\otimes \proj{1}_{\O_2}. 
    \end{align}
\end{remark}

We explicitly construct $\mathcal{F}_{(x_0,x_1)}$ with Cliffords ${\normalfont{\texttt{CL}}}$, completely dephasing channels ${\mathcal{D}}$ and completely depolarising channels ${\mathcal{P}}$ in Figure~\ref{fig:sot_from_clifford}. 

\begin{figure}[!htpb]
    \centering
    \begin{subfigure}[t]{0.6\textwidth}
        \centering
        \includegraphics[scale=0.4]{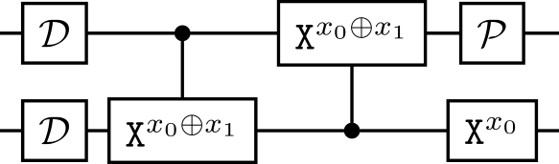}
        \caption{The compact form of $\mathcal{F}_{(x_0,x_1)}$ }
	    \label{fig:F_x}
    \end{subfigure}
    
    \begin{subfigure}[t]{0.3\textwidth}
        \centering
        \includegraphics[scale=0.4]{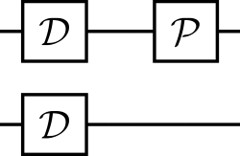}
        \caption{The explicit form of $\mathcal{F}_{(0,0)}$}
        \label{fig:F_00}
    \end{subfigure}
    \begin{subfigure}[t]{0.3\textwidth}
        \centering
        \includegraphics[scale=0.4]{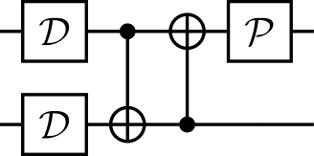}
        \caption{The explicit form of $\mathcal{F}_{(0,1)}$}
        \label{fig:F_01}
    \end{subfigure}
    
    \begin{subfigure}[t]{0.3\textwidth}
        \centering
        \includegraphics[scale=0.4]{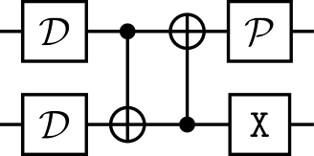}
        \caption{The explicit form of $\mathcal{F}_{(1,0)}$}
        \label{fig:F_10}
    \end{subfigure}
    \begin{subfigure}[t]{0.3\textwidth}
        \centering
        \includegraphics[scale=0.4]{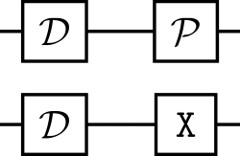}
        \caption{The explicit form of $\mathcal{F}_{(1,1)}$}
        \label{fig:F_11}
    \end{subfigure}
    \caption{(a) The compact form of $\mathcal{F}_{(x_0,x_1)}$ and (b)--(e) the explicit form of $\mathcal{F}_{(0,0)}$, $\mathcal{F}_{(0,1)}$, $\mathcal{F}_{(1,0)}$, and $\mathcal{F}_{(1,1)}$ of strong oblivious transfer channels. 
    We apply the two $\CNOT$ gates if and only if $x_0\neq x_1$ and apply the $\XGate$ gate if and only if $x_0=1$ in both (a) and (b)--(e). Thus (a) is a compact form of (b)--(e). }
    \label{fig:sot_from_clifford}
\end{figure}

The completely dephasing channel $\mathcal{D}$ can further be constructed by applying $\mathbbm{I}_2$ and $\ZGate$ uniformly at random. The completely depolarising channel $\mathcal{P}$ can be constructed by applying $\mathbbm{I}_2$, $\ZGate$, $\XGate$ and $\XGate\ZGate$ uniformly at random. That motivates us to define the set of strong oblivious transfer Cliffords. Equivalently, strong oblivious transfer Cliffords can realize strong oblivious transfer channels. 
\begin{definition}[Strong oblivious transfer Cliffords]
    Strong oblivious transfer Cliffords  $\mathcal{F}_{(x_0,x_1)}^{(r_0,r_1,r_2,r_3)}\in \CL$, where $(x_0,x_1,r_0,r_1,r_2,r_3)\in\{0,1\}^6$, are given by circuits in Figure~\ref{fig:F_xr}.
\end{definition}
\begin{figure}[!htpb]
    \centering
    \includegraphics[scale=0.4]{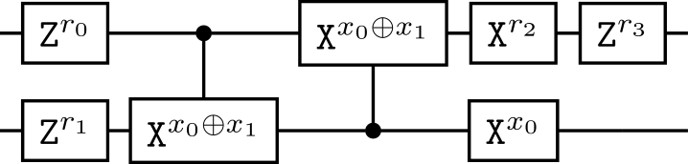}
    \caption{Strong oblivious transfer Cliffords $\mathcal{F}_{(x_0,x_1)}^{(r_0,r_1,r_2,r_3)}$. }
    \label{fig:F_xr}
\end{figure}

\begin{lemma}\label{lemma:sot_from_clifford}
    A quantum homomorphic encryption protocol that allows to delegate $\mathcal{F}_{(x_0,x_1)}^{(r_0,r_1,r_2,r_3)}$ and with $\epsilon$-correctness, $\epsilon_d$-data privacy, $\epsilon_c$-circuit privacy can simulate a quantum homomorphic encryption protocol that allows to delegate $\mathcal{F}_{(x_0,x_1)}$
    with $\epsilon$-correctness, $\epsilon_d$-data privacy and $\epsilon_c$-circuit privacy.  
\end{lemma}
\begin{proof}
    Suppose that there is a quantum homomorphic encryption protocol $Q$ that delegates $\mathcal{F}_{(x_0,x_1)}^{(r_0,r_1,r_2,r_3)}$ with $\epsilon$-correctness, $\epsilon_d$-data privacy, $\epsilon_c$-circuit privacy. The quantum homomorphic encryption protocol $Q'$ that delegates $\mathcal{F}_{(x_0,x_1)}$ can be constructed from $Q$ as follows. 
    The key generation, encryption and decryption maps of $Q'$ are the same as that of $Q'$. The evaluation map of $Q'$ requires further randomization, i.e. Bob fixes $(x_0,x_1)$, generates $(r_0, r_1, r_2, r_3)$ uniformly at random and evaluates $\mathcal{F}_{(x_0,x_1)}^{(r_0,r_1,r_2,r_3)}$. 
    Namely,
    \begin{align}\label{eqn:unencrypted_unitaries_and_channels}
        \mathcal{F}_{(x_0,x_1)}=\frac{1}{16}\sum_{r_0,r_1,r_2,r_3\in\{0,1\}}\mathcal{F}_{(x_0,x_1)}^{(r_0,r_1,r_2,r_3)}, 
    \end{align}
    and hence
    \begin{align}\label{eqn:encrypted_unitaries_and_channels}
        \hat{\mathcal{F}}_{(x_0,x_1)}=\frac{1}{16}\sum_{r_0,r_1,r_2,r_3\in\{0,1\}}\hat{\mathcal{F}}_{(x_0,x_1)}^{(r_0,r_1,r_2,r_3)}.
    \end{align}
    The correctness of $Q$ implies that for any $\mathcal{F}_{(x_0,x_1)}^{(r_0,r_1,r_2,r_3)}$, any $\ket{\psi'}\in{\A\R_{\A}}$ and any key $k=\KeyGen(\kappa)$
    \begin{align}
        \epsilon\geq\Delta\left(\Dec_k\hat{\mathcal{F}}_{(x_0,x_1)}^{(r_0,r_1,r_2,r_3)}\Enc_k(\proj{\psi'}_{\A\R_{\A}}),\mathcal{F}_{(x_0,x_1)}^{(r_0,r_1,r_2,r_3)}(\proj{\psi'}_{\A\R_{\A}}) \right).
    \end{align}
    Averaging over the uniform distribution of $(r_0,r_1,r_2,r_3)$, applying the convexity of the trace distance in~\cite[Theorem~9.3]{Nielsen_2010} and substituting Eq.~\eqref{eqn:unencrypted_unitaries_and_channels} and Eq.~\eqref{eqn:encrypted_unitaries_and_channels}, we conclude that for any $\mathcal{F}_{(x_0,x_1)}$, any $\ket{\psi'}\in{\A\R_{\A}}$ and any key $k=\KeyGen(\kappa)$ 
    \begin{align}
        \epsilon\geq\Delta\left(\Dec_k \hat{\mathcal{F}}_{(x_0,x_1)}\Enc_k(\proj{\psi'}_{\A\R_{\A}}),\mathcal{F}_{(x_0,x_1)}(\proj{\psi'}_{\A\R_{\A}}) \right), 
    \end{align}
    which is exactly the correctness of $Q'$.
    
    The data privacy of $Q$ translates directly to the data privacy of $Q'$, because both the key generation and encryption maps are identical. That is, for any $\rho\in\mathscr{S}(\A)$
    \begin{align}
        \epsilon_d \geq \Delta\left(\mathbbm{E}(\Enc_k(\rho)), \mathbbm{E}(\Enc_k(\rho'))\right), 
    \end{align}
    for both $Q$ and $Q'$. 
    
    The circuit privacy of $Q$ requires that for any $\ket{\psi}\in {\hat{\A}\R_{\hat{\A}}}$, there must exist $\ket{\psi'}\in{\A\R_{\A}}$ and $\mathcal{N}\in {\rm CPTP}(\A\R_{\hat{\A}},\hat{\A}\R_{\hat{\A}})$ such that for any $\mathcal{F}_{(x_0,x_1)}^{(r_0,r_1,r_2,r_3)} $
    \begin{align}
        \epsilon_c\geq  \Delta\left(\hat{\mathcal{F}}_{(x_0,x_1)}^{(r_0,r_1,r_2,r_3)}[\proj{\psi}_{\hat{\A}\R_{\hat{\A}}}],\mathcal{N}(\mathcal{F}_{(x_0,x_1)}^{(r_0,r_1,r_2,r_3)}(\proj{\psi'}_{\A\R_{\A}}))\right). 
    \end{align}
    Similar to the technique we use for the correctness, averaging over the uniform distribution of $(r_0,r_1,r_2,r_3)$, applying the convexity of the trace distance in~\cite[Theorem~9.3]{Nielsen_2010} and substituting Eq.~\eqref{eqn:unencrypted_unitaries_and_channels} and Eq.~\eqref{eqn:encrypted_unitaries_and_channels}, we conclude that for any $\ket{\psi}\in {\hat{\A}\R_{\hat{\A}}}$, there must exist $\ket{\psi'}\in{\A\R_{\A}}$ and $\mathcal{N}\in {\rm CPTP}(\A\R_{\hat{\A}},\hat{\A}\R_{\hat{\A}})$ such that for any $\mathcal{F}_{(x_0,x_1)}$
    \begin{align}
        \epsilon_c\geq  \Delta\left(\hat{\mathcal{F}}_{(x_0,x_1)} [\proj{\psi}_{\hat{\A}\R_{\hat{\A}}}],\mathcal{N}(\mathcal{F}_{(x_0,x_1)}(\proj{\psi'}_{\A\R_{\A}}))\right).
    \end{align}
    This shows that $Q'$ has $\epsilon_c$-circuit privacy.

\end{proof}
\begin{remark}
    We can reduce quantum homomorphic encryption which allows the delegation of $\mathcal{F}_{(x_0,x_1)}$ to quantum homomorphic encryption which allows the delegation of $\mathcal{F}_{(x_0,x_1)}^{(r_0,r_1,r_2,r_3)}$. However, the converse is not necessarily true. Quantum homomorphic encryption allowing $\mathcal{F}_{(x_0,x_1)}$ is weaker than quantum homomorphic encryption allowing $\mathcal{F}_{(x_0,x_1)}^{(r_0,r_1,r_2,r_3)}$. 
\end{remark}
In the following part, we will prove that one can construct a standard oblivious transfer protocol by a quantum homomorphic encryption protocol with $\{\mathcal{F}_{(x_0,x_1)}\}_{x_0,x_1\in\{0,1\}}\subset \mathscr{F}$. The $\epsilon$-correctness, $\epsilon_d$-data privacy and $\epsilon_c$-circuit privacy of quantum homomorphic encryption translates to the $\epsilon$-completeness, $\frac{1}{2}(1+\epsilon_d)$-soundness against a cheating Bob and $\frac{1}{2}+\epsilon_c$-soundness against a cheating Alice of standard oblivious transfer.

Now we clarify once more the notations before we state the lemmas. Recall the notations in the actual protocol. Alice's message is $\sigma\in\mathscr{S}({\hat{\A}})$ and the purification of Alice's message is $\ket{\psi}\in {\hat{\A}\R_{\hat{\A}}}$. Bob's program is $\proj{x_0,x_1}_{\B}\in\mathscr{S}(\B) $. Bob's channel is $\hat{\mathcal{F}}_{(x_0,x_1)}\in {\rm CPTP} ({\hat{\A}}, {\hat{\O}})$ which maps Alice's message $\sigma\in \mathscr{S}({\hat{\A}})$ to Bob's message $\hat{\mathcal{F}}_{(x_0,x_1)}(\sigma)\in\mathscr{S}({\hat{\O}})$. Alice obtains $\hat{\mathcal{F}}_{(x_0,x_1)}(\proj{\psi}_{\hat{\A}\R_{\hat{\A}}})\in\mathscr{S}({\hat{\O} \R_{\hat{\A}}})$. 

Recall the notations in the ideal protocol. Alice's input is $\rho'\in \mathscr{S}({\A})$ and the purification of Alice's input is  $\ket{\psi'}\in {\A \R_{\A}}$. Charlie's channel is $\mathcal{F}_{(x_0,x_1)}\in {\rm CPTP}({\A},{\O})$ which maps Alice's input $\rho'$ to Charlie's message $\mathcal{F}_{(x_0,x_1)}(\rho')$. Alice post-processes $\mathcal{F}_{(x_0,x_1)}(\proj{\psi'}_{\A\R_{\A}})\in\mathscr{S}({\O \R_{\A}})$ with $\mathcal{N}\in {\rm CPTP}({\O \R_{\A}}, {\hat{\O} \R_{\hat{\A}}})$ and obtains $\mathcal{N}(\mathcal{F}_{(x_0,x_1)}(\proj{\psi'}_{\A\R_{\A}}))\in \mathscr{S}({\hat{\O} \R_{\hat{\A}}})$.

We sketch the proof of the reduction from quantum oblivious transfer to quantum homomorphic encryption. In Protocol~\ref{prot:ot_from_qhe}, we show how one can realize the standard oblivious transfer protocol using quantum homomorphic encryption. 
In Lemma~\ref{lemma:completeness}, we translate the $\epsilon$-correctness of quantum homomorphic encryption to the $\epsilon$-completeness of standard oblivious transfer. In Lemma~\ref{lemma:soundness_against_bob}, we translate the $\epsilon_d$-data privacy to the $\frac{1}{2}(1+\epsilon_d)$-soundness against a cheating Bob. In Lemma~\ref{lemma:soundness_against_alice}, we translate the $\epsilon_c$-circuit privacy of a quantum homomorphic encryption protocol to the $(\frac{1}{2}+\epsilon_c)$-soundness of a quantum oblivious transfer protocol against a cheating Alice. We piece these lemmas together in Theorem~\ref{thm:construction}, where we show how a quantum homomorphic encryption protocol can apply a standard oblivious transfer protocol while taking into account the properties of completeness and soundness. 

In Protocol~\ref{prot:ot_from_qhe}, we reduce standard oblivious transfer to quantum homomorphic encryption. Alice inputs $i\in\{0,1\}$ and Bob inputs $(x_0,x_1)\in\{0,1\}^2$ in standard oblivious transfer. Now Alice and Bob make use of a black-box quantum homomorphic encryption protocol $(\mathscr{F},\KeyGen, \Enc, \Eval, \Dec)$ with $\{\mathcal{F}_{(x_0,x_1)}\}_{x_0,x_1\in\{0,1\}}\subset \mathscr{F}$. Alice generates Alice's key $k$ with $\KeyGen$ and prepares Alice's input state $\proj{i,0}_{\A}\in \mathscr{S}({\A})$. Alice then encrypts and sends Alice's message $\sigma =\Enc_k(\proj{i, 0}_{\A})\in\mathscr{S}({\hat{\A}})$ to Bob. Bob evaluates and sends Bob's message $\theta=\hat{\mathcal{F}}_{(x_0,x_1)}(\sigma)\in \mathscr{S}({\hat{\O}})$ back to Alice. Alice decrypts and obtains Alice's output state $\tau=\Dec_k(\theta)\in \mathscr{S}(\O)$. Finally, Alice performs the optimal positive operator valued measure $\{M_{\hat{x}}\}_{\hat{x}\in\{0,1\}}$ on $\O$ to guess $\hat{x}$ for $x_i$ and outputs $\hat{x}$.

\begin{algorithm}
    \caption{Applying standard oblivious transfer with quantum homomorphic encryption}
    \label{prot:ot_from_qhe} 
    \begin{algorithmic}[1]
        \Require $i\in \{0,1\}$, $(x_0,x_1) \in \{0,1\}^2$. \Comment{Alice's and Bob's inputs}
        \Ensure ${\rm output}_A\in \{0,1\}$. \Comment{Alice's output. }
        \State Alice: $k\leftarrow \KeyGen$. 
        \State Alice: Encode $\sigma\leftarrow \Enc_k(\proj{i,0}_{\A})$ and send it to Bob. 
        \State Bob: Evaluate $\theta\leftarrow \mathcal{F}_{(x_0,x_1)}(\sigma)$ and send it to Alice. 
        \State Alice: Decrypt $\tau \leftarrow \Dec_k(\theta)$. 
        \State Alice: Measure $\{M_{x}\}_{x \in \{0,1\}}$ with outcome $\hat{x}$. \Comment{Alice measures to determine $\hat{x}$. }
        \State Alice: ${\rm output}_A\leftarrow \hat{x}$. 
    \end{algorithmic}
\end{algorithm}

\begin{lemma}\label{lemma:completeness}
    If the quantum homomorphic encryption in Protocol~\ref{prot:ot_from_qhe} has $\epsilon$-correctness, 
    then the constructed standard oblivious transfer has $\epsilon$-completeness.
\end{lemma}
\begin{proof}
    From Definition~\ref{def:standard_ot}, the correctness of the standard oblivious transfer requires that the output is approximately correct if both Alice and Bob are honest. Note that for any $\ket{i,0}\in {\A}$ and any $(x_0, x_1)\in \{0,1\}^2$, we have 
    \begin{align}
        \mathcal{F}_{(x_0,x_1)}(\proj{i,0}_{\A})= \frac{\mathbbm{I}_2}{2}\otimes \proj{x_{i}}_{\O_2}. 
    \end{align}
    Due to the $\epsilon$-correctness of the quantum homomorphic encryption protocol defined in Definition~\ref{def:correctness}, we have 
    \begin{align}
        \Delta\left(\Dec_k\hat{\mathcal{F}}_{(x_0,x_1)}\Enc_k[\proj{i,0}_{\A}] ,\frac{\mathbbm{I}_2}{2}\otimes\proj{x_{i}}_{\O_2}\right)\leq \epsilon. 
    \end{align}
    According to~\cite[Eq.~(9.22)]{Nielsen_2010}, the trace distance can be written as
    \begin{align}
         & \Delta\left(\Dec_k\hat{\mathcal{F}}_{(x_0,x_1)}\Enc_k[\proj{i,0}_{\A}] ,\frac{\mathbbm{I}_2}{2}\otimes\proj{x_{i}}_{\O_2}\right) \nonumber\\
         & = \sup_{0\leq Q\leq \mathbbm{I}_{4}} \Tr \left[Q\left(\frac{\mathbbm{I}_2}{2}\otimes\proj{x_{i}}_{\O_2} - \Dec_k\hat{\mathcal{F}}_{(x_0,x_1)}\Enc_k[\proj{i,0}_{\A}] \right)\right]. 
    \end{align}
    Let $Q=\mathbbm{I}_2\otimes\proj{x_{i}}_{\O_2}$, we obtain 
    \begin{align}
         & \Delta\left(\Dec_k\hat{\mathcal{F}}_{(x_0,x_1)}\Enc_k[\proj{i,0}_{\A}] ,\frac{\mathbbm{I}}{2}\otimes\proj{x_{i}}_{\O_2}\right) \nonumber\\
         & \geq 1- \Tr\left(\frac{\mathbbm{I}}{2}\otimes\proj{x_{i}}_{\O_2} \Dec_k\hat{\mathcal{F}}_{(x_0,x_1)}\Enc_k[\proj{i,0}_{\A}] \right). 
    \end{align}
    Suppose that Alice performs the positive operator valued measure $\{\proj{0}_{\O},\proj{1}_{\O}\}$ and obtains the measurement outcome $\hat{x}'$. Thus
    \begin{align}
        \Pr[({\hat{x}'=x_i})]=\Tr\left(\frac{\mathbbm{I}}{2}\otimes\proj{x_{i}}_{\O_2}
        \Dec_k\hat{\mathcal{F}}_{(x_0,x_1)}\Enc_k[\proj{i,0}_{\A}
        ] \right)\geq 1-\epsilon. 
    \end{align}
    Since the optimal measurement performs better than the above measurement, we have 
    \begin{align}
        \Pr[\hat{x}=x_i] \geq \Pr[\hat{x}'=x_i] \geq 1-\epsilon.
    \end{align}
    Hence the standard oblivious transfer protocol in Definition~\ref{def:standard_ot} has $\epsilon$-completeness. 
\end{proof}

\begin{lemma}\label{lemma:soundness_against_bob}
    In Protocol~\ref{prot:ot_from_qhe}, if the quantum homomorphic encryption protocol has $\epsilon_d$-data privacy, then the constructed standard oblivious transfer protocol has $\frac{1}{2}(1+\epsilon_d)$-soundness against Bob. 
\end{lemma}
\begin{proof}
    Suppose that the quantum homomorphic encryption protocol has $\epsilon_d$-data privacy. 
    Consider an honest Alice and a malicious Bob. Suppose that Alice's input is $\proj{i,0}_{\A}$. Because Bob does not know Alice's key, Bob perceives Alice's message as $\mathbbm{E}(\Enc_k[\proj{i,0}_{\A}])$ where the expectation is taken over a uniform distribution over all keys $k$. 
    From the $\epsilon_d$-data privacy of the quantum homomorphic encryption protocol defined in Definition~\ref{def:data_privacy}, we have
    \begin{align}
        \Delta\left(\mathbbm{E}(\Enc_k[\proj{0,0}_{\A}]),\mathbbm{E}(\Enc_k[\proj{1,0}_{\A}])\right)\leq \epsilon_d. 
    \end{align}
    Now Bob measures Alice's message and guesses $\hat{i}$ for $i$. Using the Holevo-Helstrom theorem in~\cite{Helstrom_1967,Holevo_1973} (or in a textbook, e.g.~\cite[Theorem 3.4]{Watrous_2018}), the cheating probability of Bob, i.e., the probability that Bob guesses the given state correctly when Bob is given either of two states, each with a probability of $\frac{1}{2}$, is at most 
    \begin{align}
        \Pr[\hat{i}=i] \leq \frac{1}{2} (1+\epsilon_d). 
    \end{align}
    Thus, the standard oblivious transfer protocol in Definition~\ref{def:standard_ot} has $\frac{1}{2} (1+\epsilon_d)$-soundness against a cheating Bob. 
\end{proof}

Before we proceed to Lemma~\ref{lemma:soundness_against_alice}, we prove two more claims. In Claim~\ref{claim:chosen_basis}, we show that the Schmidt basis of Alice's input can be chosen as the computational basis in Protocol~\ref{prot:ot_from_qhe} in the ideal protocol. 
\begin{claim}\label{claim:chosen_basis}
    Consider Protocol~\ref{prot:ot_from_qhe} in the ideal protocol. Suppose that the input $\ket{\psi'}\in {\A \R_{\A}}$ is of the form 
    \begin{align}\label{eqn:psi_prime}
        \ket{\psi'} =\sum_{i,i',j,j'\in \{0,1\}} a_{ii'jj'} \ket{i,i'}_{\A}\ket{j,j'}_{\R_{\A}},  
    \end{align}
    and $\mathcal{N}\in {\rm CPTP}({\O \R_{\A}}, {\hat{\O} \R_{\hat{\A}}})$ is any channel. Then there exists an input $\ket{\psi''}\in {\A \R_{\A}}$ of the form 
    \begin{align}\label{eqn:psi_double_prime}
        \ket{\psi''} = \sum_{i,i'\in \{0,1\}} \sqrt{p_{ii'}}\ket{i,i'}_{\A}\ket{i,i'}_{\R_{\A}} ,
    \end{align}
    and a post-processing channel $\mathcal{N}'\in {\rm CPTP}({\O \R_{\A}}, {\hat{\O} \R_{\hat{\A}}})$ such that for any $x_0,x_1 = 0,1$,
    \begin{align}
        \mathcal{N}\left(\mathcal{F}_{(x_0,x_1)}(\proj{\psi'}_{\A\R_{\A}})\right)=\mathcal{N}'\left(\mathcal{F}_{(x_0,x_1)}(\proj{\psi''}_{\A \R_{\A}})\right). 
    \end{align}
\end{claim}
\begin{proof}
    Here, we construct $\mathcal{N}'$ explicitly. 
    Consider an input $\ket{\psi'}\in {\A \R_{\A}}$ of the form Eq.~\eqref{eqn:psi_prime}. 
    Using the definition of $ \mathcal{F}_{(x_0,x_1)}$ in Definition~\ref{def:sot_channels}, we have 
    \begin{align}
        \mathcal{F}_{(x_0,x_1)}(\proj{\psi'}_{\A\R_{\A}})=\sum_{i,i'\in \{0,1\}}p_{ii'}\frac{\mathbbm{I}_2}{2}\otimes\proj{x_{i\oplus i'}\oplus i'}_{\O_2}\otimes \rho_{ii'}, 
    \end{align}
    where $\tau_{ii'}= \Tr_{\A}(\proj{i,i'}_{\A} \proj{\psi'}_{\A\R_{\A}})\in \mathscr{S}({\R_{\A}})$, 
    and $p_{ii'} = \Tr(\tau_{ii'})$ and $\rho_{ii'} = \frac{\tau_{ii'}}{p_{ii'}}$. Consider an input $\ket{\psi''}\in {\A \R_{\A}}$ of the form Eq.~\eqref{eqn:psi_double_prime}. Hence, \begin{align}\label{eqn:ideal_output}
        \mathcal{F}_{(x_0,x_1)}(\proj{\psi''})=\sum_{i,i'\in \{0,1\}}p_{ii'}\frac{\mathbbm{I}_2}{2}\otimes\proj{x_{i\oplus i'}\oplus i'}_{\O_2} \otimes \proj{i,i'}_{\R_{\A}}. 
    \end{align}
    By observation, we can construct $\mathcal{N}''\in {\rm CPTP}(\O \R_{\A},\O \R_{\A}) $ such that 
    \begin{align}
        \mathcal{N}''(\rho) = \sum_{i,i'\in\{0,1\}}\Tr_{\R_{\A}}(\proj{i,i'}_{\R_{\A}}\rho) \otimes \rho_{ii'},
    \end{align}
    which satisfies
    \begin{align}\label{eqn:state}
        \mathcal{F}_{(x_0,x_1)}(\proj{\psi'}_{\A\R_{\A}}) =\mathcal{N}''(\mathcal{F}_{(x_0,x_1)}(\proj{\psi''}_{\A\R_{\A}})).
    \end{align}
    Hence we can construct $\mathcal{N}' = \mathcal{N}\circ \mathcal{N}''$ which satisfies 
    \begin{align}\label{eqn:simulation}
        \mathcal{N}(\mathcal{F}_{(x_0,x_1)}(\proj{\psi'}_{\A\R_{\A}}))=\mathcal{N}'(\mathcal{F}_{(x_0,x_1)}(\proj{\psi''}_{\A\R_{\A}})).
    \end{align}
    Therefore, for any input $\ket{\psi'}\in {\A \R_{\A}}$ and any channel $\mathcal{N}\in {\rm CPTP}({\O \R_{\A}}, {\hat{\O} \R_{\hat{\A}}})$, there exist $\ket{\psi''}\in {\A \R_{\A}}$ and $\mathcal{N}'\in {\rm CPTP}({\O \R_{\A}}, {\hat{\O} \R_{\hat{\A}}})$ such that for any $\mathcal{F}_{(x_0,x_1)}$, Eq.~\eqref{eqn:simulation} holds.
\end{proof}

In Claim~\ref{claim:correct_chance} we show that the cheating probability of Alice is at most $\frac{1}{2}$ in the ideal protocol if the Schmidt basis of Alice's input is the computational basis. 
\begin{claim}\label{claim:correct_chance}
    Consider Protocol~\ref{prot:ot_from_qhe} in the ideal protocol. For any input $\ket{\psi''}\in {\A \R_{\A}}$ of the form Eq.~\eqref{eqn:psi_double_prime}, any channel $\mathcal{N}'\in {\rm CPTP}({\O \R_{\A}}, {\hat{\O} \R_{\hat{\A}}})$, and any positive operator valued measure $\{M_{(\hat{x}_0,\hat{x}_1)}\}_{\hat{x}_0,\hat{x}_1\in\{0,1\}}$ on ${\hat{\O}  \R_{\hat{\A}}}$ where $(\hat{x}_0,\hat{x}_1)\in \{0,1\}^2$ are Alice's guesses for $(x_0,x_1)$, the probability that Alice guesses correctly satisfies
    \begin{align}
        \Pr[(\hat{x}_0,\hat{x}_1)=(x_0,x_1)] \leq \frac{1}{2} ,
    \end{align}
    when $(x_0, x_1)$ are chosen uniformly at random.
\end{claim}
\begin{proof}
    Suppose that in the ideal protocol, Alice inputs $\ket{\psi''}\in {\A \R_{\A}}$ of the form Eq.~\eqref{eqn:psi_double_prime}. Alice applies $\mathcal{N}'\in {\rm CPTP}({\O \R_{\A}}, {\hat{\O} \R_{\hat{\A}}})$. Let $\B$ denotes Bob's register. The definition of soundness against Alice requires Bob to input $(x_0,x_1)$ uniformly at random, or equivalently to input $\frac{1}{4}\sum_{x_0,x_1}\proj{x_0,x_1}_{\B}$. The state of both Alice and Bob in the ideal protocol is 
    \begin{align}\label{eqn:random_program_ideal_output}
        \rho_{\rm ideal} = \frac{1}{4}\sum_{x_0,x_1\in\{0,1\}} \mathcal{N'}(\mathcal{F}_{(x_0,x_1)}(\proj{\psi''}_{\A \R_{\A}}))\otimes\proj{x_0,x_1}_{\B}. 
    \end{align}
    Suppose that Alice measures $\{M_{(\hat{x}_0,\hat{x}_1)}\}_{\hat{x}_0,\hat{x}_1\in\{0,1\}}$ and guesses $(\hat{x}_0,\hat{x}_1)$ for $(x_0,x_1)$. The probability that Alice guesses correctly is
    \begin{align}\label{eqn:ideal_correct}
        \Pr[{(\hat{x}_0,\hat{x}_1)} = {(x_0,x_1)}] &= \sum_{x_0,x_1\in\{0,1\}}\Tr\left(M_{(x_0,x_1)}\otimes\proj{x_0,x_1}_{\B} \rho_{\rm ideal}\right) \\
        &= \frac{1}{4}\sum_{x_0,x_1\in\{0,1\}}\Tr\left(M_{(x_0,x_1)}\mathcal{N}'(\mathcal{F}_{(x_0,x_1)}(\proj{\psi''}_{\A \R_{\A}}))\right) 
        ,
    \end{align}
    where we substituted Eq.~\eqref{eqn:random_program_ideal_output} and took the trace over $\B$. Let $\mathcal{N}'^*$ be the adjoint of $\mathcal{N}'$ and $M_{(x_0,x_1)}' = \mathcal{N}'^*(M_{(x_0,x_1)})$. $\mathcal{N}'$ is trace preserving, thus $\mathcal{N}'^*$ is unital. Since $\{M_{(x_0,x_1)}\}_{x_0,x_1\in\{0,1\}}$ is a positive operator valued measure on $\hat{\O}\R_{\hat{\A}}$, $\{M_{(x_0,x_1)}\}_{x_0,x_1\in\{0,1\}}$ is also a positive operator valued measure on $\O\R_{\A}$. By replacing $\mathcal{N}'$ and $\{M_{(x_0,x_1)}\}_{x_0,x_1\in\{0,1\}}$ with $\mathcal{N}'^*$ and $\{M_{(x_0,x_1)}'\}_{x_0,x_1\in\{0,1\}}$, we obtain
    \begin{align}\label{eqn:ideal_correct_modified}
        \Pr[{(\hat{x}_0,\hat{x}_1)}={(x_0,x_1)}]=\frac{1}{4}\sum_{x_0,x_1\in\{0,1\}}\Tr\left(M_{(x_0,x_1)}'\mathcal{F}_{(x_0,x_1)}(\proj{\psi''}_{\A \R_{\A}})\right). 
    \end{align}
    Recall that $\ket{\psi''}$ has the form Eq.~\eqref{eqn:psi_double_prime} and thus  $\mathcal{F}_{(x_0,x_1)}(\proj{\psi''}_{\A\R_{\A}})$ has the form Eq.~\eqref{eqn:ideal_output}. Substituting both Eq.~\eqref{eqn:psi_double_prime} and Eq.~\eqref{eqn:ideal_output} into Eq.~\eqref{eqn:ideal_correct_modified}, we obtain
    \begin{align}\label{eqn:ideal_correct_simplified}
        & \Pr[(\hat{x}_0,\hat{x}_1)=(x_0,x_1)|\rho_{\rm ideal}]\nonumber \\
        & = \frac{1}{4}\sum_{x_0,x_1,i,i'\in\{0,1\}} \Tr \left(p_{ii'}M_{(x_0,x_1)}'\frac{\mathbbm{I}_2}{2}\otimes\proj{x_{i\oplus i'}\oplus i'}_{\O_2} \otimes \proj{i,i'}_{\R_{\A}} \right). 
    \end{align}
    To upper bound the probability, we replace $\proj{x_{i\oplus i'}\oplus i'}_{\O_2}$ with $\mathbbm{I}_{2}$. Substituting $\sum_{x_0,x_1\in\{0,1\}}M_{(x_0,x_1)}'= \mathbbm{I}_{\O\R_{\A}}$ as well as $\sum_{i,i'\in\{0,1\}}p_{ii'}=1$, we obtain that the probability that Alice guesses correctly is at most $\frac{1}{2}$ in the ideal protocol
    \begin{equation}\label{eqn:ideal_correct_final}
        \Pr[(\hat{x}_0,\hat{x}_1)=(x_0,x_1)|\rho_{\rm ideal}]\leq \frac{1}{2}\sum_{i,i'\in\{0,1\}}p_{ii'} = \frac{1}{2}, 
    \end{equation}
    which completes the proof. 

\end{proof}
With Claim~\ref{claim:chosen_basis} and Claim~\ref{claim:correct_chance} in hand, we are ready to prove the bound on soundness against Alice in Lemma~\ref{lemma:soundness_against_alice}. The circuit privacy bounds the trace distance between Alice's state in the ideal protocol and that in the actual protocol. The trace distance further bounds the difference between Alice's cheating probability in the actual model and the ideal model. Note that Alice's cheating probability in the ideal protocol is $\frac{1}{2}$. Hence we can upper-bound Alice's cheating probability in the actual protocol. 
\begin{lemma}\label{lemma:soundness_against_alice}
    If the quantum homomorphic encryption protocol In Protocol~\ref{prot:ot_from_qhe} has $\epsilon_c$-circuit privacy, 
    then the constructed standard oblivious transfer protocol has $(\frac{1}{2}+\epsilon_c)$-soundness against Alice. 
\end{lemma}
\begin{proof}
    Suppose that in the actual protocol, Alice inputs $\ket{\psi}\in {\hat{\A}\R_{\hat{\A}}}$. Again Let $\B$ denote Bob's register. The definition of soundness against Alice requires Bob to input $(x_0,x_1)$ uniformly at random, or equivalently Bob to input $\frac{1}{4}\sum_{x_0,x_1}\proj{x_0,x_1}_{\B}$. Bob applies $\hat{\mathcal{F}}_{(x_0,x_1)}$ according to his values for $(x_0,x_1)$. Alice and Bob's joint state in the actual protocol is 
    \begin{align}\label{eqn:random_program_actual_output}
        \rho_{\rm actual} = \frac{1}{4}\sum_{x_0,x_1\in\{0,1\}}\hat{\mathcal{F}}_{(x_0,x_1)}[\proj{\psi}_{\hat{\A}\R_{\hat{\A}}}]\otimes \proj{x_0,x_1}_{\B}.
    \end{align}
    Suppose that Alice measures $\{M_{(\hat{x}_0,\hat{x}_1)}\}_{\hat{x}_0,\hat{x}_1\in\{0,1\}}$ and guesses $(\hat{x}_0,\hat{x}_1)$ for $(x_0,x_1)$. 
    The probability that Alice can guess correctly is
    \begin{align}\label{eqn:actual_correct}
        \Pr[{(\hat{x}_0,\hat{x}_1)}={(x_0,x_1)}| \rho_{\rm actual}]=\sum_{{x_0,x_1\in\{0,1\}}}\Tr \left(M_{(x_0,x_1)}\otimes \proj{x_0,x_1}_{\B} \rho_{\rm actual}\right).
    \end{align}
    The absolute value of Eq.~\eqref{eqn:ideal_correct} minus Eq.~\eqref{eqn:actual_correct} can be bounded by the trace distance between $\rho_{\rm ideal}$ and $\rho_{\rm actual}$, as shown in the Holevo-Helstrom theorem in~\cite{Helstrom_1967,Holevo_1973} (or a modern version in~\cite[Theorem 3.4]{Watrous_2018})
    \begin{align}\label{eqn:difference_between_ideal_and_actual}
        \left|\Pr[(\hat{x}_0,\hat{x}_1)=(x_0,x_1)|\rho_{\rm actual}]-\Pr[{(\hat{x}_0,\hat{x}_1)}=(x_0,x_1)|\rho_{\rm ideal}]\right|\leq \Delta(\rho_{\rm actual},\rho_{\rm ideal}). 
    \end{align}
    Substituting Eq.~\eqref{eqn:random_program_ideal_output} and Eq.~\eqref{eqn:random_program_actual_output} into Eq.~\eqref{eqn:difference_between_ideal_and_actual} and applying the convexity of the trace distance in~\cite[Theorem~9.3]{Nielsen_2010}, we obtain
    \begin{align}
        \Delta\left(\rho_{\rm actual},\rho_{\rm ideal}\right)\leq \frac{1}{4}\sum_{x_0,x_1\in\{0,1\}}\Delta\left(\hat{\mathcal{F}}_{(x_0,x_1)}[\proj{\psi}_{\hat{\A}\R_{\hat{\A}}}],\mathcal{N'}\left(\mathcal{F}_{(x_0,x_1)}(\proj{\psi''}_{\A \R_{\A}})\right)\right) \leq \epsilon_c ,
    \end{align}
    where the last inequality is due to the $\epsilon_c$-circuit privacy of the quantum homomorphic encryption protocol defined in Definition~\ref{def:circuit_privacy}, which implies that for any $\ket{\psi}$ we can find $\ket{\psi''}$ and $\mathcal{N}'$ such that for any $(x_0,x_1)$
    \begin{align}
        \Delta\left(\hat{\mathcal{F}}_{(x_0,x_1)}[\proj{\psi}_{\hat{\A}\R_{\hat{\A}}}],\mathcal{N'}\left(\mathcal{F}_{(x_0,x_1)}(\proj{\psi''}_{\A \R_{\A}})\right)\right)\leq \epsilon_c. 
    \end{align}
    Thus,
    \begin{align}\label{eqn:difference_between_ideal_and_actual_final}
        \left|\Pr[{(\hat{x}_0,\hat{x}_1)}={(x_0,x_1)}|\rho_{\rm actual}]-\Pr[{(\hat{x}_0,\hat{x}_1)}={(x_0,x_1)}|\rho_{\rm ideal}]\right|\leq \epsilon_c. 
    \end{align}
    Substituting Eq.~\eqref{eqn:ideal_correct_final} in Claim~\ref{claim:correct_chance} into Eq.~\eqref{eqn:difference_between_ideal_and_actual_final}, we obtain that the cheating probability of Alice is at most 
    \begin{align}
        \Pr[{(\hat{x}_0,\hat{x}_1)}={(x_0,x_1)}|\rho_{\rm actual}]\leq \frac{1}{2}+\epsilon_c.  
    \end{align}
    Hence, the oblivious transfer protocol in Definition~\ref{def:standard_ot} has $(\frac{1}{2}+\epsilon_c)$-soundness. 
\end{proof}
We can immediately prove Theorem~\ref{thm:construction} with Protocol~\ref{prot:ot_from_qhe}, Lemma~\ref{lemma:completeness}, Lemma~\ref{lemma:soundness_against_bob} and Lemma~\ref{lemma:soundness_against_alice}. 
\begin{theorem}\label{thm:construction}
    Suppose that there is a quantum homomorphic encryption protocol $Q$ with $\epsilon$-correctness, $\epsilon_d$-data privacy and $\epsilon_c$-circuit privacy that delegates $\mathcal{F}_{(x_0,x_1)}$. Then one can use $Q$ to execute standard oblivious transfer with $\epsilon$-completeness, $(\frac{1}{2}+\epsilon_c)$-soundness against a cheating Alice and $\frac{1}{2}(1+\epsilon_d)$-soundness against a cheating Bob. 
\end{theorem}
\begin{proof}
    Suppose that we apply Protocol~\ref{prot:ot_from_qhe} with a quantum homomorphic encryption protocol with $\epsilon$-correctness, $\epsilon_d$-data privacy and $\epsilon_c$-circuit privacy. Lemma~\ref{lemma:completeness} shows that the $\epsilon$-completeness of the standard oblivious transfer is satisfied, Lemma~\ref{lemma:soundness_against_bob} shows that the $\frac{1}{2}(1+\epsilon_d)$-soundness against a cheating Bob is satisfied, and Lemma~\ref{lemma:soundness_against_alice} shows that the $\frac{1}{2}+\epsilon_c$-soundness against a cheating Alice is satisfied, which completes the proof. 
\end{proof}

\section{Bounds for quantum homomorphic encryption}\label{sec:bounds_for_qhe}

\subsection{Lower bounds for quantum homomorphic encryption}\label{sec:lower_bounds_for_qhe}

Here, we prove a lower bound for a broad family of quantum homomorphic encryption protocols. The reduction in Theorem~\ref{thm:construction} can translate a bound on quantum oblivious transfer to the bound on quantum homomorphic encryption that can delegate $\mathcal{F}_{(x_0,x_1)}$. Lemma~\ref{lemma:sot_from_clifford} shows that quantum homomorphic encryption allowing Cliffords can simulate quantum homomorphic encryption allowing $\mathcal{F}_{(x_0,x_1)}$. Combining Theorem~\ref{thm:construction}, Theorem~\ref{thm:srot_bound} and Lemma~\ref{lemma:sot_from_clifford}, we obtain a lower bound for quantum homomorphic encryption with Cliffords.

\begin{corollary}\label{cor:lower_bound_for_qhe}
    A quantum homomorphic encryption protocol with $\epsilon$-correctness, $\epsilon_d$-data privacy and $\epsilon_c$-circuit privacy that allows the delegated computation of $\mathcal{F}_{(x_0,x_1)}$ satisfies
    \begin{align}
        \epsilon_c + \epsilon_d + 4\sqrt{\epsilon} \geq \frac{1}{2}. 
    \end{align}
    Moreover, the above bound also applies to quantum homomorphic encryption allowing $\mathcal{F}_{(x_0,x_1)}^{(r_0,r_1,r_2,r_3)}$ or with Cliffords. 
\end{corollary}

\subsection{Upper bounds for quantum homomorphic encryption}\label{sec:upper_bounds_for_qhe}

In this section, we relate the circuit privacy of quantum homomorphic encryption that delegates a set of channels to the maximal error probability of multi-channel quantum hypothesis testing over the set of channels. In this way, we can prove upper bounds on circuit privacy of quantum homomorphic encryption. 

We first introduce the multi-channel quantum hypothesis testing problem. Suppose that we are given a quantum channel $\mathcal{F}$ which is chosen from a set $\mathscr{F}$ of quantum channels. We apply the quantum channel $\mathcal{F}$ on an input $\ket{\psi'}\in{\A\R_{\A}}$, perform a positive operator valued measure $\{M_{\mathcal{F}'}\}_{\mathcal{F}'\in\mathscr{F}}$ on the output $\mathcal{F}(\proj{\psi'}_{\A\R_{\A}})\in\mathscr{S}(\O\R_{\A})$ and use the measurement outcome $\mathcal{F}'$ to guess for $\mathcal{F}$. The maximal error probability of multi-channel quantum hypothesis testing over $\mathscr{F}$ is 
\begin{align}\label{eqn:maximal_error_probability}
    p_{\max}(\mathscr{F}) = \min_{\ket{\psi'}\in {\A\R_{\A}},\{M_{\mathcal{F}'}\}_{\mathcal{F}'\in\mathscr{F}}}\max_{\mathcal{F}\in\mathscr{F}} \Tr\left((\mathbbm{I}-M_{\mathcal{F}})\mathcal{F}(\proj{\psi'}_{\A\R_{\A}})\right). 
\end{align}

Now, consider the circuit privacy of quantum homomorphic encryption that allows the delegation of a set of channels $\mathscr{F}$. Suppose that Alice is malicious and Bob is honest. Recall the actual protocol in Figure~\ref{fig:qhe_cheating_alice} and the ideal protocol in Figure~\ref{fig:qhe_ideal}. In the actual protocol and the ideal protocol, Alice's final state is $\hat{\mathcal{F}}(\proj{\psi}_{\hat{\A}\R_{\hat{\A}}})\in \mathscr{S}({ \hat{\O}\R_{\hat{\A}}})$ and $\mathcal{F}(\proj{\psi'}_{\A \R_{\A}})\in \mathscr{S}({\O \R_{\A}})$ respectively. 
The circuit privacy quantifies how well Alice in the ideal protocol can simulate Alice in the actual protocol. There is always a possible strategy for the simulation: Alice in the ideal protocol performs multi-channel quantum hypothesis testing and then uses the measurement outcome to simulate Alice in the actual protocol. More precisely, first, Alice in the ideal protocol sends $\ket{\psi'}\in{\A\R_{\A}}$ to Charlie, obtains $\mathcal{F}(\proj{\psi'}_{\A\R_{\A}})\in \mathscr{S}(\O\R_{\A})$, measures $\{M_{\mathcal{F}'}\}_{\mathcal{F}'\in\mathscr{F}}$ on $\mathcal{F}(\proj{\psi'}_{\A\R_{\A}})$ and guesses $\mathcal{F}'$ for $\mathcal{F}$; second, Alice in the ideal protocol use $\hat{\mathcal{F}}'(\proj{\psi}_{\hat{\A}\R_{\hat{\A}}})$ to simulate $\hat{\mathcal{F}}(\proj{\psi}_{\hat{\A}\R_{\hat{\A}}})$. We can then use this strategy to upper bound the circuit privacy of quantum homomorphic encryption. 

\begin{theorem}\label{thm:general_upper_bound}
    The circuit privacy of a quantum homomorphic encryption protocol that delegates a set $\mathscr{F}$ of channels is upper bounded by the maximal error probability of multi-channel quantum hypothesis testing over a set $\mathscr{F}$ of channels
    \begin{align}
        \epsilon_c\leq p_{\max}(\mathscr{F}). 
    \end{align}
\end{theorem}
\begin{proof}
    Suppose that Alice in the actual protocol obtains $\hat{\mathcal{F}}(\proj{\psi}_{\hat{\A}\R_{\hat{\A}}})$ and that Alice in the ideal protocol obtains $\mathcal{F}(\proj{\psi'}_{\A\R_{\A}})$. Alice in the ideal protocol performs the multi-channel quantum hypothesis testing and uses the measurement outcome to simulate Alice in the actual protocol. That is, Alice in the ideal protocol sends a state $\ket{\psi'}\in{\A\R_{\A}}$ to Charlie, obtains $\mathcal{F}(\proj{\psi'}_{\A\R_{\A}})\in \mathscr{S}(\O\R_{\A})$, measures the positive operator valued measurement $\{M_{\mathcal{F}'}\}_{\mathcal{F}'\in\mathscr{F}}$, obtains the measurement outcome $\mathcal{F}'$ and then constructs $\hat{\mathcal{F}}(\proj{\psi}_{\hat{\A}\hat{\R}_{\hat{\A}}})$. The corresponding quantum channel $\mathcal{N}\in{\rm CPTP}(\O\R_{\A},\hat{\O}\R_{\hat{A}})$ is given by
    \begin{align}\label{eqn:mathcal_N}
        \mathcal{N}(\rho)= \sum_{\mathcal{F}'\in\mathscr{F}} \Tr(M_{\mathcal{F}'} \rho) \hat{\mathcal{F}'}(\proj{\psi}_{\hat{\A}\R_{\hat{\A}}}). 
    \end{align}
    Substituting Eq.~\eqref{eqn:mathcal_N} into Eq.~\eqref{eqn:circuit_privacy}, applying the convexity of the trace distance in~\cite[Theorem~9.3]{Nielsen_2010} and note that trace distance is upper-bounded by $1$, we obtain
    \begin{align}\label{eqn:trace_distance_hypothesis_testing}
        \Delta(\hat{\mathcal{F}}(\proj{\psi}_{\hat{\A}\R_{\hat{\A}}}),\mathcal{N}(\mathcal{F}(\proj{\psi'}_{\A\R_{\A}})))
        \leq \Tr((\mathbbm{I}-M_{\mathcal{F}})\mathcal{F}(\proj{\psi'}_{\A\R_{\A}})),  
    \end{align}
    which holds for all possible $\mathcal{F}\in\mathscr{F}$, $\ket{\psi'}\in {\A\R_{\A}}$,  $\{M_{\mathcal{F}'}\}_{\mathcal{F}'\in\mathscr{F}}$, and $\ket{\psi}\in{\hat{\A}\R_{\hat{\A}}}$. We now maximize over all possible $\mathcal{F}\in\mathscr{F}$, minimize over all possible $\ket{\psi'}\in {\A\R_{\A}}$ and $\{M_{\mathcal{F}'}\}_{\mathcal{F}'\in\mathscr{F}}$, and maximize over all possible $\ket{\psi}\in{\hat{\A}\R_{\hat{\A}}}$. Recall Remark~\ref{rmk:circuit_privacy} and Eq.~\eqref{eqn:maximal_error_probability}, we obtain
    \begin{align}
        \epsilon_c\leq p_{\max}(\mathscr{F}), 
    \end{align}
    which completes the proof. 
\end{proof}
By applying a trivial positive operator valued measure $\{M_{\mathcal{F}}=\frac{1}{|\mathscr{F}|}\mathbbm{I}_{\O\R_{\A}}\}_{\mathcal{F}\in\mathscr{F}}$, we can immediately obtain a trivial upper bound $\epsilon_c\leq 1-\frac{1}{|\mathscr{F}|}$ for an arbitrary $\mathscr{F}$. We can further derive a better upper bound for $\mathscr{F}=\{\mathcal{F}_{(x_0,x_1)}\}_{x_0,x_1\in\{0,1\}}$. 
\begin{corollary}\label{cor:sot_upper_bound}
    For any quantum homomorphic encryption which only allows the delegated computation of $\mathcal{F}_{(x_0,x_1)} $, the circuit privacy is bounded by 
    \begin{align}
        \epsilon_c\leq \frac{1}{2}. 
    \end{align}
\end{corollary}
\begin{proof}
    Let $\ket{\psi'}=\ket{0,0}\ket{0,0}$ and $\{M_{(x_0,x_1)}=\mathbbm{I}_2\otimes\frac{1}{2}\proj{x_0}_{\O_2}\}_{_0,x_1\in\{0,1\}}$. Hence 
    \begin{align}
        \mathcal{F}_{(x_0,x_1)}(\proj{\psi'}_{\A\R_{\A}})=\frac{\mathbbm{I}_2}{2}\otimes \proj{x_0}_{\O_2} \otimes \proj{0,0}_{\R_{\A}}.
    \end{align}
    Thus we have 
    \begin{equation}
        p_{\max}(\mathscr{F})\leq\max_{x_0,x_1\in\{0,1\}}\Tr\left((\mathbbm{I}_4-M_{(x_0,x_1)})\mathcal{F}_{(x_0,x_1)}(\proj{\psi'}_{\A\R_{\A}})\right)=\frac{1}{2}. 
    \end{equation}
    We then immediately obtain Corollary~\ref{cor:sot_upper_bound} from Theorem~\ref{thm:general_upper_bound}. 
\end{proof}

\section{Conclusion}

In conclusion, we formally define quantum homomorphic encryption and its information-theoretic circuit privacy, correctness and data privacy. We then reduce quantum oblivious transfer to quantum homomorphic encryption with strong oblivious transfer channels. We show our reduction works for a broad class of quantum homomorphic encryption by further reducing quantum homomorphic encryption with strong oblivious transfer channels to quantum homomorphic encryption with Cliffords. Combining our reduction and a lower bound for quantum oblivious transfer, we obtain a lower bound for a broad class of quantum homomorphic encryption
\begin{align}
    \epsilon_d+\epsilon_c+ 4\sqrt{\epsilon} \geq \frac{1}{2}. 
\end{align}
To be complete, we also prove an upper bound for quantum homomorphic encryption by constructing a simulation strategy for Alice from multi-channel quantum hypothesis testing. We show that the circuit privacy of quantum homomorphic encryption that delegates a set $\mathscr{F}$ of channels is upper bounded by the maximal error probability of multi-channel quantum hypothesis testing
\begin{align}
    \epsilon_c\leq p_{\max}(\mathscr{F}). 
\end{align}
As a corollary, we find $\epsilon_c \leq \frac{1}{2}$ for quantum homomorphic encryption that only delegates $\mathcal{F}_{(x_0,x_1)}$.

In Figure~\ref{fig:qhe_bound}, we present lower bounds for quantum homomorphic encryption protocols that only allow the delegation of $\mathcal{F}_{(x_0,x_1)}$ with perfect correctness. Each point in Figure~\ref{fig:qhe_bound} is denoted by $(\epsilon_d,\epsilon_c)$. The line corresponds to the lower bound, and the shaded area indicates the impossible region. The diamond point $(0,\frac{1}{2})$ can be achieved using~\cite{Ouyang_2018} asymptotically. The square point $(1,0)$ can be trivially achieved by Alice sending her input to Bob without encryption. It is still unknown whether the line is reachable at points other than the diamond point. In particular, we do not have quantum homomorphic encryption protocols which trade some data privacy for some circuit privacy, and thus we cannot achieve other points on the line without additional resources. (If Alice and Bob have shared randomness or can use weak coin flipping, then interpolation would be possible.) 

We note that the impossible region could potentially be enlarged if we consider more powerful function classes so that, for example, $1$-out-of-$n$ oblivious transfer can be reduced to homomorphic encryption, reducing the probability that Alice and Bob correctly guess their opponent's inputs by chance. We leave this for future work.

\begin{figure}[!htpb]
    \centering
    \includegraphics[scale=0.65]{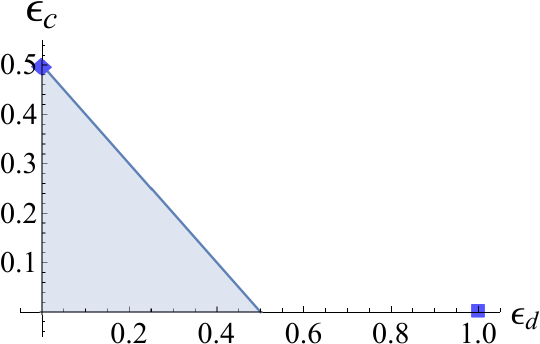}
    \caption{The lower bound for quantum homomorphic encryption that only allows to delegate $\mathcal{F}_{(x_0,x_1)}$ with perfect correctness. Each point is denoted by $(\epsilon_d,\epsilon_c)$. The line corresponds to the lower bound. The diamond point $(0,\frac{1}{2})$ and the square point $(1,0)$ are known to be achievable. The shaded area is impossible.}
    \label{fig:qhe_bound}
\end{figure} 


Our definitions and techniques can also apply to reductions from other cryptographic primitives (e.g. secure multi-party computation~\cite{Lo_1997,Colbeck_2007}, coin-flipping~\cite{Mochon_2007,Chailloux_2009,Aharonov_2016,Miller_2020}, bit commitment~\cite{Lo_and_Chau_1997,Mayers_1997}) to quantum homomorphic encryption, which may allow better privacy and correctness trade-offs.

\section{Acknowledgements}\label{sec:acknow}
We thank Christopher Portmann for useful discussions on security definitions. MT, YO and YH are supported by 
the National Research Foundation, Singapore and A*STAR under its CQT Bridging Grant. They are also funded by the Quantum Engineering programme grant NRF2021-QEP2-01-P06.

\appendix
\section{Equivalence between standard oblivious transfer and semi-random oblivious transfer}\label{sec:equivalence_between_ot_and_srot}
In this subsection, we give the detailed proof for Theorem~\ref{thm:equivalence_between_ot_and_srot}. 
\begin{proof}
    First, we reduce semi-random oblivious transfer to standard oblivious transfer, as is shown in Protocol~\ref{prot:semi_random_ot_from_standard_ot}. Bob inputs two data bits $(x_0,x_1)\in\{0,1\}^2$ in semi-random oblivious transfer. Alice generates $i\in\{0,1\}$ uniformly at random. Then Alice and Bob perform standard oblivious transfer in which Alice inputs $i$ and Bob inputs $(x_0,x_1)$ and from which Alice obtains $A$ and Bob obtains $B$. Each party accepts in semi-random oblivious transfer if and only if she or he accepts in standard oblivious transfer. If Alice accepts, Alice's output is ${\rm output}_A=(i,A)$.

    \begin{algorithm}[!htbp]
        \caption{Semi-random oblivious transfer from standard oblivious transfer}\label{prot:semi_random_ot_from_standard_ot}
        \begin{algorithmic}[1]
            \Require $(x_0,x_1)\in\{0,1\}^2$. \Comment{Bob's input. }
            \Ensure ${\rm output}_A\in\{0,1\}^2\cup\{\Abort\}$, ${\rm output}_B\in\{\Accept,\Abort\}$. \Comment{Alice's and Bob's outputs. }
            \State Alice: Generate $i\leftarrow \$ $. \Comment{Alice generates a bit uniformly at random}
            \State Alice and Bob: Perform $(A,B)\leftarrow {\rm StandardOT}(i,(x_0,x_1))$. 
            \If{$A= {\Abort}$}
                \State Alice: ${\rm output}_A\leftarrow{\Abort}$. 
            \Else
                \State Alice: ${\rm output}_A\leftarrow (i,A)$. \Comment{Alice accepts in {\rm SemirandomOT} iff Alice accepts in {\rm StandardOT}. } 
            \EndIf
            \State Bob: ${\rm output}_B\leftarrow B$. \Comment{Bob accepts in {\rm SemirandomOT} iff Bob accepts in {\rm StandardOT}. } 
        \end{algorithmic}
    \end{algorithm}
    Now we show that the $\delta$-completeness, $P_A^\star$-soundness against Alice and $P_B^\star$-soundness against Bob also translate from standard oblivious transfer to semi-random oblivious transfer in the reduction. 
    \begin{itemize}
        \item Completeness: Suppose that both Alice and Bob are honest. Then $i$ is uniformly at random. Due to the $\delta$-completeness of standard oblivious transfer, both parties accept, and with a probability of at least $1-\delta$, $y=x_{i}$. Hence the $\delta$-completeness of semi-random oblivious transfer is satisfied. 
        \item Soundness against a cheating Alice: Suppose that Alice is malicious while Bob is honest. Due to the $P_A^\star$-soundness against a cheating Alice of standard oblivious transfer, with a probability of at most $P_A^\star$, Alice can guess Bob's $(x_0,x_1)$ correctly and both parties accept. Therefore, the $P_A^\star$-soundness against a cheating Alice of semi-random oblivious transfer is satisfied. 
        \item Soundness against a cheating Bob: Suppose that Alice is honest while Bob is malicious. Due to the $P_B^\star$-soundness against a cheating Alice of standard oblivious transfer, with a probability of at most $P_B^\star$, Bob can guess Alice's $i$ correctly and both parties accept. Thus, the $P_B^\star$-soundness against a cheating Alice of semi-random oblivious transfer is satisfied. 
    \end{itemize}

    Second, we reduce standard oblivious transfer to semi-random oblivious transfer, as is shown in Protocol~\ref{prot:standard_ot_from_semi_random_ot}. Alice inputs $i\in\{0,1\}$ and Bob inputs $(x_0,x_1)\in\{0,1\}^2$ in standard oblivious transfer. Bob generates $(y_0,y_1)\in\{0,1\}^2$ uniformly at random. Then Alice and Bob perform semi-random oblivious transfer in which Bob inputs $(y_0,y_1)$ and from which Alice obtains $A$ and Bob obtains $B$. They then declare if they have aborted to the other party. Each party accepts in standard oblivious transfer if and only if both accept in semi-random oblivious transfer. If both accept, Alice sets $(j,\hat{y})=A$. Then Alice encrypts and sends the request $r=i\oplus j$ to Bob, Bob encrypts and sends the data bits $s_0=x_r\oplus y_0$ and $s_1=x_{\overline{r}}\oplus y_1$ to Alice. Alice decrypts the data and Alice's output is ${\rm output}_A=s_j\oplus \hat{y}$.

    \begin{algorithm}[!htbp]
        \caption{Standard oblivious transfer from semi-random oblivious transfer}\label{prot:standard_ot_from_semi_random_ot}
        \begin{algorithmic}[1]
            \Require $i\in\{0,1\}$, $(x_0,x_1)\in\{0,1\}^2$. \Comment{Alice's and Bob's inputs. }
            \Ensure ${\rm output}_B\in \{0,1\}\cup\{\Abort\}$, ${\rm output}_B\in\{\Accept,\Abort\}$. \Comment{Alice's and Bob's outputs. }
            \State Bob: Generate $y_0,y_1\leftarrow \$ $. \Comment{Bob generates two bits uniformly at random. } 
            \State Alice and Bob: Perform $(A,B)\leftarrow {\rm SemirandomOT}(y_0,y_1)$. 
            \State Alice and Bob: Declare if they have aborted to the other party. 
            \If{$(A= {\Abort})\lor(B={\Abort})$}
                \State Alice: ${\rm output}_A\leftarrow \Abort$. 
                \State Bob: ${\rm output}_B\leftarrow \Abort$. 
            \Else
                \State Alice: $(j,\hat{y})\leftarrow A$. 
                \State Alice: Send $r\leftarrow i\oplus j$ to Bob. 
                \State Bob: Send $s_0\leftarrow x_r\oplus y_0$ and $s_1\leftarrow x_{\overline{r}}\oplus y_1$ to Alice. 
                \State Alice: ${\rm output}_A\leftarrow s_j\oplus \hat{y}$. \Comment{Alice accepts in {\rm StandardOT} iff both accept in {\rm SemirandomOT}. }
                \State Bob: ${\rm output}_B\leftarrow {\Accept}$. \Comment{Bob accepts in {\rm StandardOT} iff both accept in {\rm SemirandomOT}. }
            \EndIf
        \end{algorithmic}
    \end{algorithm}
    
    Now we show that the $\delta$-completeness, $P_A^\star$-soundness against Alice and $P_B^\star$-soundness against Bob also translate from semi-random oblivious transfer to standard oblivious transfer in the reduction. 
    \begin{itemize}
        \item Completeness: Suppose that both Alice and Bob are honest. Due to the $\delta$-completeness of semi-random oblivious transfer, both parties accept, and with a probability of at least $1-\delta$, $\hat{y}=y_{j} $. Considering that $\hat{x}=x_{i}\oplus y_{j}\oplus \hat{y}$, thus with a probability of at least $1-\delta$, $\hat{x}=x_{i}$. Therefore, the $\delta$-completeness of standard oblivious transfer is satisfied. 
        \item Soundness against a cheating Alice: Suppose that Alice is malicious and Bob is honest. Due to the $P_A^\star$-soundness against a cheating Alice of semi-random oblivious transfer, with a probability of at most $P_A^\star$, Alice can guess Bob's $(y_0,y_1)$ and both parties accept. Since $y_0y_1$ is uniformly random and $x_r=s_0\oplus y_0$ and $ x_{\overline{r}})=s_1\oplus y_1$, with a probability of at most $P_A^\star$, Alice can guess Bob's $(x_0,x_1)$ and both parties accept. Hence, the $P_A^\star$-soundness against a cheating Alice of standard oblivious transfer is satisfied. 
        \item Soundness against a cheating Bob: Suppose that Alice is honest and Bob is malicious. Due to the $P_B^\star$-soundness against a cheating Bob of semi-random oblivious transfer, with a probability of at most $P_0^\star$, Bob can guess Alice's $j$ and both parties accept. Considering that $j$ is uniformly random and that $i=r\oplus j$, with a probability of at most $P_B^\star$, Bob can guess Alice's $i$ and both parties accept. Therefore, the $P_B^\star$-soundness against a cheating Bob of standard oblivious transfer is satisfied. 
    \end{itemize}
\end{proof}


\bibliographystyle{quantum}
\bibliography{quantumref}

\begin{thebibliography}{10}

\bibitem{Fitzsimons_2017}
Joseph~F Fitzsimons.
\newblock ``Private quantum computation: an introduction to blind quantum
  computing and related protocols''.
\newblock \href{https://doi.org/10.1038/s41534-017-0025-3}{npj Quantum
  Information {\bf 3}, 1--11}~(2017).

\bibitem{Aharonov_2008}
Dorit Aharonov, Michael Ben-Or, and Elad Eban.
\newblock ``Interactive proofs for quantum computations''~(2008)
  \href{http://doi.org/10.48550/arXiv.0810.5375}{arXiv:0810.5375}.

\bibitem{Broadbent_2009}
Anne Broadbent, Joseph Fitzsimons, and Elham Kashefi.
\newblock ``Universal blind quantum computation''.
\newblock In 2009 50th Annual IEEE Symposium on Foundations of Computer
  Science.
\newblock \href{https://doi.org/10.1109/FOCS.2009.36}{Pages 517--526}.
\newblock ~(2009).

\bibitem{Morimae_2013}
Tomoyuki Morimae and Keisuke Fujii.
\newblock ``Blind quantum computation protocol in which alice only makes
  measurements''.
\newblock \href{https://doi.org/10.1103/PhysRevA.87.050301}{Phys. Rev. A {\bf
  87}, 050301}~(2013).

\bibitem{Reichardt_2013}
Ben~W Reichardt, Falk Unger, and Umesh Vazirani.
\newblock ``Classical command of quantum systems''.
\newblock \href{https://doi.org/10.1038/nature12035}{Nature {\bf 496},
  456--460}~(2013).

\bibitem{Mantri_2017}
Atul Mantri, Tommaso~F. Demarie, Nicolas~C. Menicucci, and Joseph~F.
  Fitzsimons.
\newblock ``Flow ambiguity: A path towards classically driven blind quantum
  computation''.
\newblock \href{https://doi.org/10.1103/PhysRevX.7.031004}{Phys. Rev. X {\bf
  7}, 031004}~(2017).

\bibitem{Yu_2014}
Li~Yu, Carlos~A. P\'erez-Delgado, and Joseph~F. Fitzsimons.
\newblock ``Limitations on information-theoretically-secure quantum homomorphic
  encryption''.
\newblock \href{https://doi.org/10.1103/PhysRevA.90.050303}{Phys. Rev. A {\bf
  90}, 050303}~(2014).

\bibitem{Broadbent_2015}
Anne Broadbent and Stacey Jeffery.
\newblock ``Quantum homomorphic encryption for circuits of low t-gate
  complexity''.
\newblock In Rosario Gennaro and Matthew Robshaw, editors, Advances in
  Cryptology -- CRYPTO 2015.
\newblock \href{https://doi.org/10.1007/978-3-662-48000-7_30}{Pages 609--629}.
\newblock Berlin, Heidelberg~(2015). Springer Berlin Heidelberg.

\bibitem{Dulek_2016}
Yfke Dulek, Christian Schaffner, and Florian Speelman.
\newblock ``Quantum homomorphic encryption for polynomial-sized circuits''.
\newblock In Matthew Robshaw and Jonathan Katz, editors, Advances in Cryptology
  -- CRYPTO 2016.
\newblock \href{https://doi.org/10.1007/978-3-662-53015-3_1}{Pages 3--32}.
\newblock Berlin, Heidelberg~(2016). Springer Berlin Heidelberg.

\bibitem{Tan_2016}
Si-Hui Tan, Joshua~A. Kettlewell, Yingkai Ouyang, Lin Chen, and Joseph~F.
  Fitzsimons.
\newblock ``A quantum approach to homomorphic encryption''.
\newblock \href{https://doi.org/10.1038/srep33467}{Scientific Reports {\bf 6},
  33467}~(2016).

\bibitem{Ouyang_2018}
Yingkai Ouyang, Si-Hui Tan, and Joseph~F. Fitzsimons.
\newblock ``Quantum homomorphic encryption from quantum codes''.
\newblock \href{https://doi.org/10.1103/PhysRevA.98.042334}{Phys. Rev. A {\bf
  98}, 042334}~(2018).

\bibitem{Mahadev_2020}
Urmila Mahadev.
\newblock ``Classical homomorphic encryption for quantum circuits''.
\newblock \href{https://doi.org/10.1137/18M1231055}{SIAM Journal on Computing
  {\bf 0}, FOCS18--189}~(2020).

\bibitem{Ouyang_2022}
Yingkai Ouyang and Peter~P. Rohde.
\newblock ``A general framework for the composition of quantum homomorphic
  encryption \& quantum error correction''~(2022)
  \href{http://doi.org/10.48550/arXiv.2204.10471}{arXiv:2204.10471}.

\bibitem{Gentry_2009}
Craig Gentry.
\newblock ``Fully homomorphic encryption using ideal lattices''.
\newblock In Proceedings of the 41st annual ACM Symposium on Theory of
  computing.
\newblock \href{https://doi.org/10.1145/1536414.1536440}{Pages 169--178}.
\newblock ~(2009).

\bibitem{Gentry_2009_thesis}
Craig Gentry.
\newblock ``A fully homomorphic encryption scheme''.
\newblock PhD thesis.
\newblock Stanford University.
\newblock ~(2009).
\newblock
  url:~\href{https://crypto.stanford.edu/craig}{crypto.stanford.edu/craig}.

\bibitem{Gentry_2010}
Craig Gentry, Shai Halevi, and Vinod Vaikuntanathan.
\newblock ``I-hop homomorphic encryption and rerandomizable yao circuits''.
\newblock In Proceedings of the 30th Annual Conference on Advances in
  Cryptology.
\newblock \href{https://doi.org/10.1007/978-3-642-14623-7_9}{Pages 155--172}.
\newblock CRYPTO'10Berlin, Heidelberg~(2010). Springer-Verlag.

\bibitem{Barak_2012}
Baoz Barak and Zvika Brakerski.
\newblock ``The swiss army knife of cryptography''~(2012)
  url:~\href{https://windowsontheory.org/2012/05/01/the-swiss-army-knife-of-cryptography/}{windowsontheory.org/2012/05/01/the-swiss-army-knife-of-cryptography/}.

\bibitem{Lindell_2017}
Yehuda Lindell.
\newblock ``Tutorials on the foundations of cryptography: Dedicated to oded
  goldreich''.
\newblock \href{https://doi.org/10.1007/978-3-319-57048-8}{Springer Publishing
  Company, Incorporated}. ~(2017).
\newblock 1st edition.

\bibitem{Esmaeilzade_2022}
Saeid Esmaeilzade, Nasrollah Pakniat, and Ziba Eslami.
\newblock ``A generic construction to build simple oblivious transfer protocols
  from homomorphic encryption schemes''.
\newblock \href{https://doi.org/10.1007/s11227-021-03826-0}{The Journal of
  Supercomputing {\bf 78}, 72--92}~(2022).

\bibitem{Reingold_2004}
Omer Reingold, Luca Trevisan, and Salil Vadhan.
\newblock ``Notions of reducibility between cryptographic primitives''.
\newblock In Moni Naor, editor, Theory of Cryptography.
\newblock \href{https://doi.org/10.1007/978-3-540-24638-1_1}{Pages 1--20}.
\newblock Berlin, Heidelberg~(2004). Springer Berlin Heidelberg.

\bibitem{Lai_2018}
Ching-Yi Lai and Kai-Min Chung.
\newblock ``On statistically-secure quantum homomorphic encryption''.
\newblock \href{https://doi.org/10.26421/QIC18.9-10-4}{Quantum Info. Comput.
  {\bf 18}, 785--794}~(2018).

\bibitem{Newman_2018}
Michael Newman.
\newblock ``Further limitations on information-theoretically secure quantum
  homomorphic encryption''~(2018)
  \href{http://doi.org/10.48550/arXiv.1809.08719}{arXiv:1809.08719}.

\bibitem{Nayak_1999}
Ashwin Nayak.
\newblock ``Optimal lower bounds for quantum automata and random access
  codes''.
\newblock In 40th Annual Symposium on Foundations of Computer Science (Cat.
  No.99CB37039).
\newblock \href{https://doi.org/10.1109/SFFCS.1999.814608}{Pages 369--376}.
\newblock ~(1999).

\bibitem{Tan_2017}
Si-Hui Tan, Yingkai Ouyang, and Peter~P. Rohde.
\newblock ``Practical somewhat-secure quantum somewhat-homomorphic encryption
  with coherent states''.
\newblock \href{https://doi.org/10.1103/PhysRevA.97.042308}{Phys. Rev. A {\bf
  97}, 042308}~(2018).

\bibitem{Ouyang_2020}
Yingkai Ouyang, Si-Hui Tan, Joseph Fitzsimons, and Peter~P. Rohde.
\newblock ``Homomorphic encryption of linear optics quantum computation on
  almost arbitrary states of light with asymptotically perfect security''.
\newblock \href{https://doi.org/10.1103/PhysRevResearch.2.013332}{Physical
  Review Research {\bf 2}, 013332}~(2020).

\bibitem{Chailloux_2013}
Andr\'{e} Chailloux, Iordanis Kerenidis, and Jamie Sikora.
\newblock ``Lower bounds for quantum oblivious transfer''.
\newblock \href{https://doi.org/10.26421/QIC13.1-2-9}{Quantum Info. Comput.
  {\bf 13}, 158--177}~(2013).

\bibitem{Chailloux_2016}
Andr\'e Chailloux and Jamie Sikora.
\newblock ``Optimal bounds for semi-honest quantum oblivious transfer''.
\newblock \href{https://doi.org/10.4086/cjtcs.2016.013}{Chicago Journal of
  Theoretical Computer Science{\bf \ 2016}}~(2016).

\bibitem{Amiri_2021}
Ryan Amiri, Robert St\'arek, David Reichmuth, Ittoop~V. Puthoor, Michal
  Mi\ifmmode~\check{c}\else \v{c}\fi{}uda, Ladislav Mi\ifmmode~\check{s}\else
  \v{s}\fi{}ta, Jr., Miloslav Du\ifmmode~\check{s}\else \v{s}\fi{}ek, Petros
  Wallden, and Erika Andersson.
\newblock ``Imperfect 1-out-of-2 quantum oblivious transfer: Bounds, a
  protocol, and its experimental implementation''.
\newblock \href{https://doi.org/10.1103/PRXQuantum.2.010335}{PRX Quantum {\bf
  2}, 010335}~(2021).

\bibitem{Audenaert_2014}
Koenraad M.~R. Audenaert and Milán Mosonyi.
\newblock ``Upper bounds on the error probabilities and asymptotic error
  exponents in quantum multiple state discrimination''.
\newblock \href{https://doi.org/10.1063/1.4898559}{Journal of Mathematical
  Physics {\bf 55}, 102201}~(2014).

\bibitem{Helstrom_1967}
Carl~W. Helstrom.
\newblock ``Detection theory and quantum mechanics''.
\newblock \href{https://doi.org/10.1016/S0019-9958(67)90302-6}{Information and
  Control {\bf 10}, 254--291}~(1967).

\bibitem{Holevo_1973}
Alexander~S. Holevo.
\newblock ``Bounds for the quantity of information transmitted by a quantum
  communication channel''.
\newblock Problems of Information Transmission {\bf 9}, 177--183~(1973).
\newblock
  url:~\href{http://mi.mathnet.ru/ppi903}{http://mi.mathnet.ru/ppi903}.

\bibitem{Watrous_2018}
John Watrous.
\newblock ``The theory of quantum information''.
\newblock \href{https://doi.org/10.1017/9781316848142}{Cambridge University
  Press}. ~(2018).

\bibitem{Fuches_1999}
C.A. Fuchs and J.~van~de Graaf.
\newblock ``Cryptographic distinguishability measures for quantum-mechanical
  states''.
\newblock \href{https://doi.org/10.1109/18.761271}{IEEE Transactions on
  Information Theory {\bf 45}, 1216--1227}~(1999).

\bibitem{Uhlmann_1976}
A.~Uhlmann.
\newblock ``The ``transition probability'' in the state space of a *-algebra''.
\newblock \href{https://doi.org/10.1016/0034-4877(76)90060-4}{Reports on
  Mathematical Physics {\bf 9}, 273--279}~(1976).

\bibitem{Nielsen_2010}
Michael~A Nielsen and Isaac Chuang.
\newblock ``Quantum computation and quantum information: 10th anniversary
  edition''.
\newblock \href{https://doi.org/10.1017/CBO9780511976667}{Cambridge University
  Press}. ~(2010).

\bibitem{Lo_1997}
Hoi-Kwong Lo.
\newblock ``Insecurity of quantum secure computations''.
\newblock \href{https://doi.org/10.1103/PhysRevA.56.1154}{Phys. Rev. A {\bf
  56}, 1154--1162}~(1997).

\bibitem{Colbeck_2007}
Roger Colbeck.
\newblock ``Impossibility of secure two-party classical computation''.
\newblock \href{https://doi.org/10.1103/PhysRevA.76.062308}{Phys. Rev. A {\bf
  76}, 062308}~(2007).

\bibitem{Mochon_2007}
Carlos Mochon.
\newblock ``Quantum weak coin flipping with arbitrarily small bias''~(2007)
  \href{http://doi.org/10.48550/arXiv.0711.4114}{arXiv:0711.4114}.

\bibitem{Chailloux_2009}
Andr{\'e} Chailloux and Iordanis Kerenidis.
\newblock ``Optimal quantum strong coin flipping''.
\newblock In 2009 50th Annual IEEE Symposium on Foundations of Computer
  Science.
\newblock \href{https://doi.org/10.1109/FOCS.2009.71}{Pages 527--533}.
\newblock IEEE~(2009).

\bibitem{Aharonov_2016}
Dorit Aharonov, André Chailloux, Maor Ganz, Iordanis Kerenidis, and Loïck
  Magnin.
\newblock ``A simpler proof of the existence of quantum weak coin flipping with
  arbitrarily small bias''.
\newblock \href{https://doi.org/10.1137/14096387X}{SIAM Journal on Computing
  {\bf 45}, 633--679}~(2016).

\bibitem{Miller_2020}
Carl~A. Miller.
\newblock ``The impossibility of efficient quantum weak coin flipping''.
\newblock In Proceedings of the 52nd Annual ACM SIGACT Symposium on Theory of
  Computing.
\newblock \href{https://doi.org/s10.1145/3357713.3384276}{Pages 916--929}.
\newblock New York, NY, USA~(2020). Association for Computing Machinery.

\bibitem{Lo_and_Chau_1997}
Hoi-Kwong Lo and H.~F. Chau.
\newblock ``Is quantum bit commitment really possible?''.
\newblock \href{https://doi.org/10.1103/PhysRevLett.78.3410}{Phys. Rev. Lett.
  {\bf 78}, 3410--3413}~(1997).

\bibitem{Mayers_1997}
Dominic Mayers.
\newblock ``Unconditionally secure quantum bit commitment is impossible''.
\newblock \href{https://doi.org/10.1103/PhysRevLett.78.3414}{Phys. Rev. Lett.
  {\bf 78}, 3414--3417}~(1997).

\end{thebibliography}

\end{document}